\documentclass[11pt]{article}

\usepackage{mathpazo}
\usepackage{amsfonts}
\usepackage{amsmath}
\usepackage{latexsym}

\usepackage[margin=1.05in]{geometry}

\usepackage[T1]{fontenc}
\usepackage{times}
\usepackage{color,graphicx}
\usepackage{array}
\usepackage{enumerate}
\usepackage{amsmath}
\usepackage{amssymb}
\usepackage{amsthm}
\usepackage{pgfplots}
\usepackage{pgf}
\usepackage{tikz}
\usepackage{filecontents}
\usepgfplotslibrary{patchplots} 
\usetikzlibrary{pgfplots.patchplots} 
\pgfplotsset{width=9cm}

\usepackage{amsthm}

\newtheorem{lemma}{Lemma}
\newtheorem{theorem}{Theorem}

\newtheorem{proposition}{Proposition}

\newtheorem{question}{Question}

\usepackage{xcolor}
\usepackage{makeidx}
\usepackage[colorlinks=true,linkcolor=blue,anchorcolor=blue,citecolor=red,urlcolor=magenta]{hyperref}

\renewcommand{\int}{\operatorname{int}}

\newcommand{\bra}[1]{\langle #1 |}
\newcommand{\ket}[1]{| #1 \rangle}

\newcommand{\ketbra}[2]{| #1 \rangle\langle #2 |}
\newcommand{\bb}[1]{\mathbb{#1}}

\newcommand\ip[2]{\ensuremath{\langle#1,#2\rangle}}

\newcommand\Tr{\mathop{\rm Tr}\nolimits}

\newcommand{\defeq}{\stackrel{\smash{\textnormal{\tiny def}}}{=}}

\renewcommand{\t}{{\scriptscriptstyle\mathsf{T}}}

\makeatletter
\def\Ddots{\mathinner{\mkern1mu\raise\p@
\vbox{\kern7\p@\hbox{.}}\mkern2mu
\raise4\p@\hbox{.}\mkern2mu\raise7\p@\hbox{.}\mkern1mu}}
\makeatother

\def\rddots{{\scriptstyle\cdot}^{{\scriptstyle\cdot}^{{\scriptstyle\cdot}}}}

\begin{document}

\title{Is Absolute Separability Determined by the Partial Transpose?}




\author{
  Srinivasan Arunachalam,\footnote{Centrum Wiskunde \& Informatica (CWI) Amsterdam, The Netherlands} 
  \quad
  Nathaniel Johnston,\footnote{Institute for Quantum Computing and Department of Combinatorics \& Optimization, University of Waterloo}
  \quad and\quad
  Vincent Russo\footnote{Institute for Quantum Computing and David R. Cheriton School of Computer Science, University of Waterloo}
}

\date{January 22, 2015}

\maketitle

\begin{abstract}
	The absolute separability problem asks for a characterization of the quantum states $\rho \in M_m\otimes M_n$ with the property that $U\rho U^\dagger$ is separable for all unitary matrices $U$. We investigate whether or not it is the case that $\rho$ is absolutely separable if and only if $U\rho U^\dagger$ has positive partial transpose for all unitary matrices $U$. In particular, we develop an easy-to-use method for showing that an entanglement witness or positive map is unable to detect entanglement in any such state, and we apply our method to many well-known separability criteria, including the range criterion, the realignment criterion, the Choi map and its generalizations, and the Breuer--Hall map. We also show that these two properties coincide for the family of isotropic states, and several eigenvalue results for entanglement witnesses are proved along the way that are of independent interest.
\end{abstract}

\section{Introduction} \label{sec:introduction}

In quantum information theory, a quantum state $\rho \in M_m \otimes M_n$ (where $M_n$ denotes the space of $n \times n$ complex matrices) is called \emph{separable} \cite{W89} if there exist constants $p_i \geq 0$ and states $\rho^{(1)}_i \in M_m$ and $\rho^{(2)}_i \in M_n$ such that $\sum_i p_i = 1$ and
\begin{align*}
	\rho = \sum_i p_i \rho^{(1)}_i \otimes \rho^{(2)}_i.
\end{align*}

Finding methods for determining whether a given quantum state is separable or \emph{entangled} (i.e., not separable) is one of the most active areas of quantum information theory research \cite{GT09,HHH09}. Although this problem is believed to be difficult in general \cite{G03,G10}, many partial results are known. For example, the \emph{positive-partial-transpose (PPT) criterion} states that if $\rho$ is separable then $(id_m \otimes T)(\rho) \geq 0$, where $\geq 0$ indicates positive semidefiniteness, $id_m : M_m \rightarrow M_m$ is the identity map, and $T : M_n \rightarrow M_n$ is the transpose map \cite{P96}. However, the converse of the PPT criterion only holds when $mn \leq 6$ \cite{S63,W76}, so additional tests for separability are required in higher dimensions.

The most natural generalization of the PPT criterion says that a state $\rho \in M_m \otimes M_n$ is separable if and only if $(id_m \otimes \Phi)(\rho)$ is positive semidefinite for all positive maps $\Phi : M_n \rightarrow M_m$ \cite{HHH96}. Thus each fixed positive $\Phi : M_n \rightarrow M_m$ gives a necessary condition for separability.

The \emph{absolute separability} problem \cite{KZ00} (sometimes called the \emph{separability from spectrum} problem \cite{OpenProb15}) asks for a characterization of the states $\rho \in M_m \otimes M_n$ with the property that $U\rho U^\dagger$ is separable for all unitary matrices $U \in M_m \otimes M_n$, which is equivalent to asking which sets of real numbers $\{\lambda_1, \lambda_2, \ldots, \lambda_{mn}\}$ are such that every state $\rho \in M_m \otimes M_n$ with eigenvalues $\lambda_1 \geq \lambda_2 \geq \cdots \geq \lambda_{mn} \geq 0$ is separable. This question was first answered in the $m = n = 2$ case in \cite{VAD01}, where it was shown that $\rho \in M_2 \otimes M_2$ is absolutely separable if and only if its eigenvalues satisfy $\lambda_1 \leq \lambda_3 + 2\sqrt{\lambda_2 \lambda_4}$, however the problem remains open in general.

One motivation for this problem comes from the fact that it is sometimes easier to determine the eigenvalues of a quantum state than it is to determine the entire structure of that state \cite{EAO02,TOKKNN13}. Thus, the absolute separability problem asks for the strongest separability test that can be devised given this restricted information. In another direction, the exact largest size of a ball of separable states centered at the maximally-mixed state $\tfrac{1}{mn}(I\otimes I) \in M_m \otimes M_n$ is known \cite{GB02}, and it is not difficult to show that every state within this ball is absolutely separable. However, there are also absolutely separable states outside of this ball, and it would be nice to have a characterization of where they are. Alternatively, we can think of states that are \emph{not} absolutely separable as those that can be used to generate entanglement when the operations at our disposal are global unitary channels \cite{GCM14}.

One approach to characterizing the states that are absolutely separable would be to instead fix some necessary test for separability and determine the set of states $\rho \in M_m \otimes M_n$ with the property that $U\rho U^\dagger$ satisfies that separability test for all unitary matrices $U \in M_m \otimes M_n$. This approach was initiated in \cite{Hil07}, where the set of states $\rho \in M_m \otimes M_n$ that are \emph{absolutely PPT} (i.e., states such that $(id_m \otimes T)(U\rho U^\dagger)$ is positive semidefinite for all unitary $U$) were completely characterized. Similarly, the very recent paper \cite{JLNR14} investigated states $\rho$ with the property that $U\rho U^\dagger$ satisfies the reduction criterion for all unitary matrices $U$ (because the reduction criterion is weaker than the partial transpose criterion, we do not explicitly consider it in the present paper). We continue this work by considering the same problem for several other separability criteria.

It was shown in \cite{Joh13} that the set of absolutely PPT states coincides with the set of absolutely separable states when $m = 2$ and $n$ is arbitrary, despite the fact that the set of PPT states is strictly larger than the set of separable states when $m = 2$ and $n \geq 4$. The question was then asked whether or not the set of absolutely PPT states and absolutely separable states coincide when $m,n \geq 3$. In the present paper, we demonstrate that several standard methods of entanglement detection are unable to answer this question.

More specifically, we introduce a general method (Lemma~\ref{lem:eigs_imp_sfs}) that can be used to show that a given entanglement witness or positive map cannot detect entanglement in any absolutely PPT state. Using this method, we prove several results of the form ``if $\rho$ is absolutely PPT, then it is also absolutely <other separability criterion>''. For example, we show that every absolutely PPT $\rho \in M_m \otimes M_n$ is also ``absolutely realignable''---i.e., $U\rho U^\dagger$ always satisfies the \emph{realignment criterion} introduced in \cite{CW03,R03}, even though there are PPT states $\rho$ that violate the realignment criterion. This difference between the usual separability problem and the absolute separability problem is illustrated in Figure~\ref{fig:ppt_from_spec}. We also prove that the absolute separability and absolute PPT properties coincide when restricted to the well-known family of isotropic states.
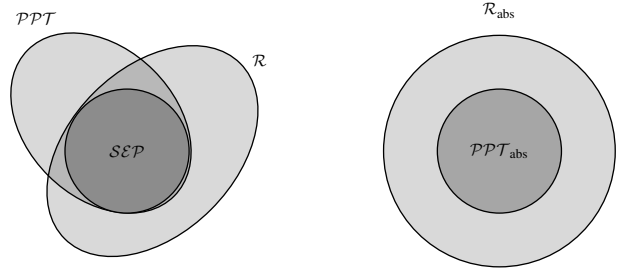
\begin{figure}[htb]
	\centering
	\begin{tikzpicture}[scale=1.1]
	  	\draw [draw=black, line width=0.5pt, fill=black!15, rotate=45] (-0.49,0.49) ellipse (1.5cm and 1cm);
	  	\draw [draw=black, line width=0.5pt, fill=black, fill opacity=0.15, rotate=-45] (-1.17,-0.69) ellipse (1.25cm and 0.9cm);
		\draw [draw=black, line width=0.5pt, fill=black!45] (-1,0) circle [radius=0.75cm];

		\draw [draw=black, line width=0.5pt, fill=black!15!white] (3.5,0) circle [radius=1.4cm];
		\draw [draw=black, line width=0.5pt, fill=black!35!white] (3.5,0) circle [radius=0.75cm];
		
		\node at (3.5,0) {\tiny{$\mathcal{PPT}_{\textup{abs}}$}};
		\node at (3.5,1.7) {\tiny{$\mathcal{R}_{\textup{abs}}$}};
		
		\node at (-1,0) {\tiny{$\mathcal{SEP}$}};
		\node at (0.6,1.1) {\tiny{$\mathcal{R}$}};
		\node at (-2.1,1.6) {\tiny{$\mathcal{PPT}$}};		

	\end{tikzpicture}
	\caption{The figure on the left represents the relationship between the set of separable states $\mathcal{SEP}$, the set of PPT states $\mathcal{PPT}$, and the set of states that satisfy the realignment criterion $\mathcal{R}$. The figure on the right represents the relationship between the set of absolutely PPT states $\mathcal{PPT}_{\textup{abs}}$ and the set $\mathcal{R}_{\textup{abs}}$ of states that are ``absolutely realignable'' (see Theorem~\ref{thm:realign_from_spec}).}\label{fig:ppt_from_spec}
\end{figure}

\section{Preliminaries} \label{sec:preliminaries}

The proofs of our results rely heavily on \emph{semidefinite programming}. Given Hermitian matrices $A \in M_n$ and $B \in M_m$ and a Hermiticity-preserving linear map $\Phi : M_n \rightarrow M_m$ (i.e., a map such that $\Phi(X^\dagger) = \Phi(X)^\dagger$ for all $X \in M_n$), the semidefinite program associated with the triple $(A,B,\Phi)$ is the following pair of optimization problems:
\begin{align*}
		\begin{matrix}
		\begin{tabular}{r l c r l}
		\multicolumn{2}{c}{\text{Primal problem}} & \quad \quad & \multicolumn{2}{c}{\text{Dual problem}} \\
		\text{maximize:} & ${\rm Tr}(AX)$ & \quad \quad & \text{minimize:} & ${\rm Tr}(BY)$ \\
		\text{subject to:} & $\Phi(X) \leq B$ & \quad \quad \quad & \text{subject to:} & $\Phi^\dagger(Y) \geq A$ \\
		\ & $X \geq 0$ & \quad \quad & \ & $Y \geq 0$,	
		\end{tabular}
		\end{matrix}
\end{align*}
where $\Phi^\dagger : M_m \rightarrow M_n$ is the \emph{dual map} of $\Phi$ defined by ${\rm Tr}(\Phi(X)Y) = {\rm Tr}(X\Phi^\dagger(Y))$ for all $X \in M_n$ and $Y \in M_m$. Semidefinite programs can be efficiently solved \cite{GLS93}, and furthermore \emph{weak duality} always holds, which tells us that ${\rm Tr}(AX) \leq {\rm Tr}(BY)$ for all feasible points $X \in M_n$ and $Y \in M_m$. In particular, this means that we can get upper bounds on the optimal value of the primal problem by simply finding a single feasible point for the dual problem (and similarly, feasible points of the primal problem give lower bounds on the optimal value of the dual problem). For a more thorough introduction to semidefinite programming, see \cite{VB94,Wat04Lec7}.

Given a linear map $\Phi : M_n \rightarrow M_m$, we recall that its \emph{Choi matrix} is the operator
\begin{align*}
	J(\Phi) \defeq n(id_n \otimes \Phi)(\ketbra{\psi^+}{\psi^+}) \in M_n \otimes M_m,
\end{align*}
where $\ket{\psi^+} = \frac{1}{\sqrt{n}}\sum_{i=1}^n\ket{i}\otimes\ket{i} \in \mathbb{C}^n \otimes \mathbb{C}^n$ is the standard maximally-entangled pure state. It is well-known that $\Phi$ is \emph{completely positive} (i.e., satisfies $(id_n \otimes \Phi)(X) \geq 0$ whenever $0 \leq X \in M_n \otimes M_n$) if and only if $J(\Phi)$ is positive semidefinite \cite{C75}.

	Our proofs will also be heavily reliant on the notion of \emph{entanglement witnesses}, which are Hermitian operators $W \in M_m \otimes M_n$ with the property that $\Tr(W \sigma) \geq 0$ for all separable $\sigma \in M_m \otimes M_n$, but $\Tr(W \rho) < 0$ for some (necessarily entangled) $\rho \in M_m \otimes M_n$. Here we say that $W$ detects the entanglement in $\rho$, and we note that every entangled $\rho$ is detected by some entanglement witness $W$. Finally, we will also make frequent use of the family of Schatten $p$-norms, defined for $p \in [1,\infty]$ by 
\begin{align*}
\| X \|_p := \left[ \Tr \left( (X^{\dagger} X)^{p/2} \right) \right]^{1/p},
\end{align*}
where we define $\|X\|_{\textup{tr}} := \|X\|_1$, $\|X\|_{F} := \|X\|_2$, and $\|X\| := \|X\|_{\infty}$ (and in these special cases, these norms are often called the \emph{trace norm}, \emph{Frobenius norm}, and \emph{operator norm}, respectively). 

The remainder of this article is organized as follows. In Section~\ref{sec:ppt}, we briefly review the characterization of states that are absolutely PPT that was originally derived in \cite{Hil07}. We then formally present the question in which we are interested in Section~\ref{sec:the-conjecture}, and briefly discuss the implications of an answer to this question. The next sections are dedicated to showing that several well-known separability criteria are unable to detect entanglement in any absolutely PPT state, and are thus unable to answer the our question. In Section~\ref{sec:special_states}, we show that for specific classes of states that absolute separability and absolute PPT coincide. Finally, in Section~\ref{sec:outlook}, we conclude and list a number of open problems and directions for future research.

\section{Absolute Positive Partial Transpose}\label{sec:ppt}

We now briefly recall some of the key points of the characterization of absolutely PPT states given in \cite{Hil07}. Indeed, the main result of that paper shows that, for each $m,n \in \mathbb{N}$, there exists a finite family of linear matrix inequalities (LMIs) with the property that $\rho \in M_m \otimes M_n$ is absolutely PPT if and only if its eigenvalues $\lambda_1 \geq \lambda_2 \geq \cdots \geq \lambda_{mn}$ satisfy each of the LMIs.

In the $m = 2$ case, the LMI that determines absolute PPT is
\begin{align*}
	L_1 := \begin{bmatrix}
		2\lambda_{2n} & \lambda_{2n-1} - \lambda_1 \\
		\lambda_{2n-1} - \lambda_1 & 2\lambda_{2n-2}
	\end{bmatrix} \geq 0,
\end{align*}
which is easily seen to be equivalent to the previously-discussed inequalities $\lambda_1 \leq \lambda_{3} + 2\sqrt{\lambda_{2} \lambda_{4}}$ when $n = 2$ and $\lambda_1 \leq \lambda_{5} + 2\sqrt{\lambda_{4} \lambda_{6}}$ when $n = 3$.

In the $m = 3$ case, there are two LMIs that determine absolute PPT:
\begin{align}\begin{split}\label{eq:ppt_from_spec_m3}
		L_1 := \begin{bmatrix}
			2\lambda_{3n} & \lambda_{3n-1}-\lambda_1 & \lambda_{3n-3}-\lambda_2 \\
			\lambda_{3n-1}-\lambda_1 & 2\lambda_{3n-2} & \lambda_{3n-4}-\lambda_3 \\
			\lambda_{3n-3}-\lambda_2 & \lambda_{3n-4}-\lambda_3 & 2\lambda_{3n-5} \\
		\end{bmatrix} \geq 0, \\
		L_2 := \begin{bmatrix}
			2\lambda_{3n} & \lambda_{3n-1}-\lambda_1 & \lambda_{3n-2}-\lambda_2 \\
			\lambda_{3n-1}-\lambda_1 & 2\lambda_{3n-3} & \lambda_{3n-4}-\lambda_3 \\
			\lambda_{3n-2}-\lambda_2 & \lambda_{3n-4}-\lambda_3 & 2\lambda_{3n-5} \\
		\end{bmatrix} \geq 0.
	\end{split}\end{align}
That is, $\rho \in M_3 \otimes M_n$ is absolutely PPT if and only if its eigenvalues satisfy both of the positive semidefiniteness conditions~\eqref{eq:ppt_from_spec_m3}.

In general, once we have fixed $m,n$ we use $L_1,L_2,L_3,\ldots$ to denote the matrices of eigenvalues whose positive semidefiniteness determine absolute PPT, and these matrices always look quite similar to the matrices~\eqref{eq:ppt_from_spec_m3} from the $m = 3$ case. For example, each $L_i$ is of size $\min\{m,n\}\times\min\{m,n\}$, the diagonal entry of each $L_i$ is $2$ times one of the $\lambda_j$'s, and each off-diagonal entry is the difference of two of the $\lambda_j$'s. Furthermore, the top-left $2\times 2$ sub-matrix of $L_1$ is always of the form
\begin{align}\label{eq:ppt_from_spec_submatrix}
	\begin{bmatrix}
		2\lambda_{mn} & \lambda_{mn-1}-\lambda_1 \\ \lambda_{mn-1}-\lambda_1 & 2\lambda_{mn-2}
	\end{bmatrix},
\end{align}
so positive semidefiniteness of~\eqref{eq:ppt_from_spec_submatrix} is a necessary (but not sufficient when $m,n \geq 3$) condition for $\rho$ to be absolutely PPT.

We note that the number of $L_i$'s that must be checked to be positive semidefinite grows exponentially in $\min\{m,n\}$ (for example, when $\min\{m,n\} = 7$ the number of $L_i$'s is $107,498$ \cite{oeisA237749}), and their exact construction is slightly complicated. However, it is not important for our purposes to be familiar with their exact construction---the properties of these matrices that we presented above are all we need.

We now present, without proof, a lemma that is well-known in matrix analysis (see, for example, \cite[Problem~III.6.14]{Bha97}).
\begin{lemma}\label{lem:unitary_min}
	Let $A,B \in M_n$ be Hermitian matrices with eigenvalues $\lambda_1\geq \cdots \geq \lambda_n$ and $\mu_1 \geq \cdots \geq \mu_n$, respectively. Then
	\begin{align*}
		\min\big\{ \Tr(AUBU^\dagger) : U \in M_n \text{ is unitary} \big\} = \sum_{j=1}^n \lambda_j \mu_{n-j+1}.
	\end{align*}
\end{lemma}

We can make use of Lemma~\ref{lem:unitary_min} to see that semidefinite programming can be used to determine whether or not a given entanglement witness is capable of detecting entanglement in an absolutely PPT state. In particular, if we have an entanglement witness $W \in M_m \otimes M_n$ with eigenvalues $\mu_1 \geq \cdots \geq \mu_{mn}$ then $W$ can detect the entanglement in some absolutely PPT state if and only if the optimal value of the following semidefinite program is strictly less than zero:
\begin{align}\begin{split}\label{sdp:witness_detect_general}
	\text{minimize:} \ \ & \sum_{j=1}^{mn} \lambda_j \mu_{mn-j+1} \\
	\text{subject to:} \ \ & L_i \geq 0 \quad \forall \, i \\
		 & \lambda_j \geq \lambda_{j+1} \geq 0 \quad \forall \, 1 \leq j \leq mn-1 \\
		 & \sum_{j=1}^{mn} \lambda_j = 1.
\end{split}\end{align}

Indeed, the constraints in the SDP~\eqref{sdp:witness_detect_general} are simply enforcing the fact that $\lambda_1 \geq \cdots \geq \lambda_{mn} \geq 0$ are the eigenvalues of some absolutely PPT state. It then follows from Lemma~\ref{lem:unitary_min} that the SDP~\eqref{sdp:witness_detect_general} computes
\begin{align*}
	\min\big\{ \Tr(W\rho) : \rho \in M_m \otimes M_n \text{ is absolutely PPT} \big\}.
\end{align*}

\section{The Absolute PPT Question} \label{sec:the-conjecture}

We now present the question that is at the heart of this work. Recall that the answer to this question was already shown to be ``yes'' in the $m = 2$ case in \cite{Joh13}.
\begin{question}\label{ques:main}
	Is it true that a quantum state $\rho \in M_m \otimes M_n$ is absolutely separable if and only if it is absolutely PPT?
\end{question}

The rest of the paper is devoted to investigating Question~\ref{ques:main}. In particular, we show that many of the standard techniques from entanglement theory cannot be used to help answer this question. We first need the following proposition.
\begin{proposition}\label{prop:rank_of_ppt_spec}
	Suppose that there exists a state $\rho \in M_m \otimes M_n$ that is absolutely PPT but not absolutely separable. Then $\rho$ has full rank.
\end{proposition}
\begin{proof}
	Suppose that $\rho$ is absolutely PPT with eigenvalues $\lambda_1 \geq \cdots \geq \lambda_{mn} = 0$ (notice that we set the smallest eigenvalue equal to $0$, so that $\rho$ does not have full rank). Our goal is to show that $\rho$ is absolutely separable.
	
	We recall from Section~\ref{sec:ppt} that the matrix~\eqref{eq:ppt_from_spec_submatrix} must be positive semidefinite. However, by using the fact that $\lambda_{mn} = 0$, we then see that $\lambda_1 = \lambda_{mn-1}$, which implies that (up to a positive scalar multiple), $\rho = I - \ketbra{v}{v}$ for some pure state $\ket{v} \in \bb{C}^m \otimes \bb{C}^n$. We now use \cite[Theorem~1]{GB02}, which says that every operator of the form $I - X$ with $\|X\|_F \leq 1$ is separable (and even absolutely separable). Since $\|\ketbra{v}{v}\|_F = 1$, it follows that $\rho$ is absolutely separable, as desired.
\end{proof}

We note that the proof of Proposition~\ref{prop:rank_of_ppt_spec} shows that the only rank-deficient absolutely PPT states are (up to normalization) the orthogonal projections of rank $nm-1$, and these states are even in the Gurvits--Barnum ball of separability.

Proposition~\ref{prop:rank_of_ppt_spec} immediately implies that the range criterion \cite{H97} for detecting entanglement, which states that the range of a separable state is spanned by product pure states, cannot possibly detect entanglement in any absolutely PPT state. To see this, simply note that the range of a full-rank state is the entire Hilbert space, which is always spanned by product states (such as the standard basis). Furthermore, Proposition~\ref{prop:rank_of_ppt_spec} also shows that most of the ``usual'' ways of creating PPT entangled states cannot possibly create absolutely PPT entangled states, since many such methods result in states that are \emph{not} of full rank (e.g., chessboard states~\cite{BP00}, states constructed by unextendible product bases~\cite{BDMSST99}, the $1$-parameter family of states constructed by the Horodeckis \cite{HHH98}, and so on). Relatively few families of bound entangled states with full rank are known \cite{BGR05,CD14}, and we have not been able to find any that are absolutely PPT (see Section~\ref{sec:upb_states}, for example).

\section{The Absolute Separablity ``Collapse''}\label{sec:collapse}

In this section, we present the main results of the paper, which show that the set of absolutely PPT states is ``closer'' to the set of absolutely separable states than the set of PPT states is to the set of separable states in the following sense: there are (many) separability criteria that are capable of detecting entanglement in PPT states, but become weaker than the PPT criterion in the ``absolute'' regime (see Figure~\ref{fig:ppt_from_spec}, for example). We already saw this for the range criterion in the previous section. We now prove that the same result holds for the realignment criterion~\cite{CW03,R03}, the Choi map~\cite{C75} and its generalizations~\cite{CKL92}, and the Breuer--Hall map~\cite{Bre06,Hal06}. That is, each of these separability criteria are incapable of detecting any entanglement in absolutely PPT states.

Before dealing with any specific separability criteria, we first need the following very important lemma, which we will make repeated use of. This lemma lets us determine that an entanglement witness cannot detect entanglement in absolutely PPT states, based only on very limited information about the eigenvalues of the witness (specifically, its largest eigenvalue and the sum of its negative eigenvalues).
\begin{lemma}\label{lem:eigs_imp_sfs}
	Let $W \in M_m \otimes M_n$ be a Hermitian operator with ${\rm Tr}(W) = 1$. Let $\mu_{1}$ be the maximum eigenvalue of $W$ and define $\ell$ to be the sum of its negative eigenvalues:
	\begin{align*}
		\ell \defeq (1 - \|W\|_{\textup{tr}})/2.
	\end{align*}
	Furthermore, define a function $f : [-\tfrac{1}{2},0] \rightarrow [\tfrac{1}{2},1]$ by:
	\begin{align*}
		f(x) \defeq \frac{1}{4}\begin{cases} \sqrt{1-4x^2}-2x+1 & \text{ if } -\tfrac{1}{2} \leq x \leq -\frac{1}{2\sqrt{2}} \\
			1+\sqrt{2} & \text{ if } -\frac{1}{2\sqrt{2}} < x < \frac{1-\sqrt{2}}{2} \\
			\sqrt{1+4x-4x^2}-2x+3 & \text{ if } \frac{1-\sqrt{2}}{2} \leq x \leq 0.
		\end{cases}
	\end{align*}
	If $\ell \geq -\tfrac{1}{2}$ and $\mu_1 \leq f(\ell)$ then ${\rm Tr}(W\rho) \geq 0$ for all absolutely PPT states $\rho \in M_m \otimes M_n$.
\end{lemma}

Before proving the lemma, we note that we have numerically found that the function $f$ described by Lemma~\ref{lem:eigs_imp_sfs} is optimal at least in the $m = n = 3$ case. That is, given any choice of $\ell$ and $\mu_1$ such that $\mu_1 > f(\ell)$, we can numerically find a Hermitian operator $W \in M_3 \otimes M_3$ and an absolutely PPT state $\rho \in M_3 \otimes M_3$ such that ${\rm Tr}(W) = 1$, $W$ has a single negative eigenvalue equal to $\ell$, the maximum eigenvalue of $W$ is $\mu_1$, and ${\rm Tr}(W\rho) < 0$.

The function $f(x)$ is plotted in Figure~\ref{fig:intervals}, where we have highlighted some important special cases. For example, $f(-\tfrac{1}{2}) = \tfrac{1}{2}$, $f(-\tfrac{2}{5}) = \tfrac{3}{5}$, $f(-\tfrac{1}{5}) = \tfrac{9}{10}$, and $f((1-\sqrt{2})/2) = (2+\sqrt{2})/4$.
\begin{figure}[htb]
	\centering
	\begin{tikzpicture}
		\begin{axis}[
			xlabel=Sum of negative eigenvalues ($\ell$),
			ylabel=Upper bound on $\mu_1$,
			grid=major,
			]
			\draw[domain=-0.5:-0.35356,smooth,variable=\t,black,line width=0.5] plot ({1000*\t+500},{250*sqrt(1-4*\t*\t)-500*\t-250});
			\draw[domain=-0.35356:-0.2071,smooth,variable=\t,black,line width=0.5] plot ({1000*\t+500},{250*(1+sqrt(2))-500});
			\draw[domain=-0.2071:0,smooth,variable=\t,black,line width=0.5] plot ({1000*\t+500},{250*sqrt(1+4*\t-4*\t*\t)-500*\t+250});

			\addplot[blue!25!black,only marks,mark=*,mark size=1pt] coordinates {(-0.5,0.5) (-0.4,0.6) (-0.2072,0.6035) (-0.3535,0.6035) (-0.2071,0.855) (-0.2,0.9) (0,1)};
			\addplot[white,only marks,mark=*,mark size=0.6pt] coordinates {(-0.2072,0.6035)};
			\node at (axis cs:0,1) [anchor=south] {{\scriptsize $(vii)$}};
			\node at (axis cs:-0.2,0.9) [anchor=south east] {{\scriptsize $(vi)$}};
			\node at (axis cs:-0.2071,0.855) [anchor=north east] {{\scriptsize $(v)$}};
			\node at (axis cs:-0.2072,0.6035) [anchor=south west] {{\scriptsize $(iv)$}};
			\node at (axis cs:-0.3535,0.6035) [anchor=south west] {{\scriptsize $(iii)$}};
			\node at (axis cs:-0.4,0.6) [anchor=south east] {{\scriptsize $(ii)$}};
			\node at (axis cs:-0.5,0.5) [anchor=east] {{\scriptsize $(i)$}};
		\end{axis}
	\end{tikzpicture}
	\caption{A plot of the upper bound $f(\ell)$ given by Lemma~\ref{lem:eigs_imp_sfs}. For example, point $(i)$ is $(-\tfrac{1}{2},\tfrac{1}{2})$, which tells us that if the sum of the negative eigenvalues ($\ell$) of $W$ equals $-\tfrac{1}{2}$ then $W$ cannot detect entanglement in absolutely PPT states if the largest eigenvalue ($\mu_1$) of $W$ is $\leq\tfrac{1}{2}$. Point $(ii)$ corresponds to $\ell = -\tfrac{2}{5}$ and $\mu_1 \leq \tfrac{3}{5}$, $(iii)$ corresponds to $\ell = -\tfrac{1}{2\sqrt{2}}$ and $\mu_1 \leq \tfrac{1}{4}(1+\sqrt{2})$, $(v)$ corresponds to $\ell = \tfrac{1}{2}(1-\sqrt{2})$ and $\mu_1 \leq \tfrac{1}{4}(2+\sqrt{2})$, $(vi)$ corresponds to $\ell = -\tfrac{1}{5}$ and $\mu_1 \leq \tfrac{9}{10}$, and $(vii)$ corresponds to $\ell = 0$ and $\mu_1 \leq 1$ (in which case the result is trivial).}\label{fig:intervals}
\end{figure}
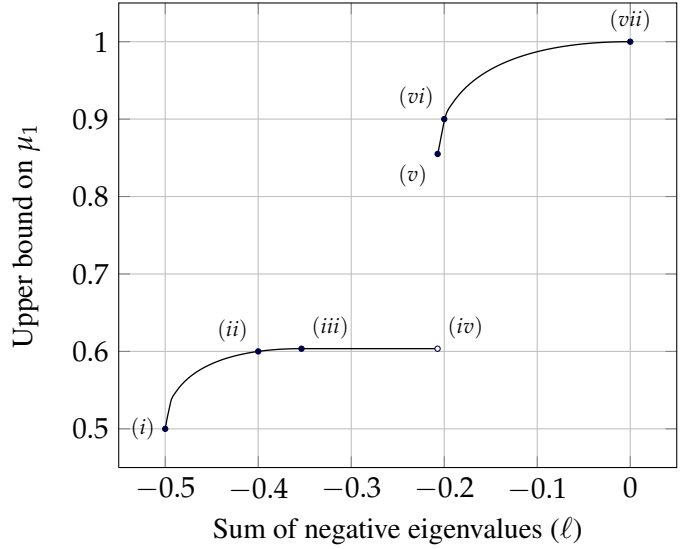
\begin{proof}[Proof of Lemma~\ref{lem:eigs_imp_sfs}]
	We prove the result by showing that the semidefinite program~\eqref{sdp:witness_detect_general} has optimal value $\geq 0$ whenever $\ell \geq -\tfrac{1}{2}$ and $\mu_1 \leq f(\ell)$. First, we replace the complicated set of LMI constraints $L_i \geq 0$ for all $i$ in this SDP with the single constraint that the $2 \times 2$ matrix~\eqref{eq:ppt_from_spec_submatrix} is positive semidefinite. Since this new SDP is a minimization problem subject to a weaker set of constraints, its optimal value is no larger than the optimal value of the SDP~\eqref{sdp:witness_detect_general}.
	
	Second, we will make some simplifying assumptions about the eigenvalues $\mu_1,\mu_2,\ldots,\mu_{mn}$. To this end, for now we fix some $\ell$ and $\lambda_1,\lambda_2,\ldots,\lambda_{mn} \geq 0$ satisfying the constraints of the SDP~\eqref{sdp:witness_detect_general}, and consider the following SDP, where we optimize over $\mu_1,\mu_2,\ldots,\mu_{mn}$:
	\begin{align}\begin{split}\label{sdp:nice_form_mu}
		\text{minimize:} \ \ & \sum_{j=1}^{mn} \lambda_j \mu_{mn-j+1} \\
		\text{subject to:} \ \ & (1 - \sum_{j=1}^{mn} |\mu_j|)/2 \geq \ell \\
				 & \mu_j \geq \mu_{j+1} \quad \forall \, 1 \leq j \leq mn-1 \\
				 & \sum_{j=1}^{mn} \mu_j = 1.
	\end{split}\end{align}
	
	It is straightforward to see that the optimal solution of the SDP~\eqref{sdp:nice_form_mu} occurs when $\mu_{mn} = \ell$ (i.e., rather than having multiple negative $\mu_j$'s, we just have one of them as negative as possible). Similarly, increasing $\mu_1$ (subject to the constraints of the SDP~\eqref{sdp:nice_form_mu}), or increasing $\mu_2$ while fixing $\mu_1$ will also decrease the value of the objective function, and similarly for increasing $\mu_{3}$ while fixing $\mu_1$ and $\mu_2$. For example, when $\ell = -2/5$ and $\mu_1 \leq f(-2/5) = 3/5$, it suffices to consider the case when $\mu_{mn} = -2/5, \mu_1 = 3/5, \mu_2 = 3/5, \mu_3 = 1/5$ ($\mu_1,\mu_2,\mu_3$ are determined by making $\mu_1$ as large as possible subject to $\mu_1 \leq f(\ell)$, then $\mu_2$ as large as possible while subject to $\mu_2 \leq \mu_1$, and so on until $\sum_i \mu_i = 1$). In general, we set $\mu_2 = \min\{\mu_1,1-\mu_1-\ell\},$ and $\mu_3 = {\rm max}\{0,1-2\mu_1-\ell\}$ (and $\mu_i = 0$ for all $4 \leq i \leq mn-1$).
	
	Since the optimal solution $\mu_1,\ldots,\mu_{mn}$ of the SDP~\eqref{sdp:nice_form_mu} satisfies the conditions described in the previous paragraph \emph{regardless} of $\lambda_1,\ldots,\lambda_{mn}$, we can assume without loss of generality in the SDP~\eqref{sdp:witness_detect_general} that $\mu_1,\ldots,\mu_{mn}$ satisfy those same conditions. That is, it suffices to show that the optimal value of the following SDP is $\geq 0$, where we recall that $\mu_1$ and $\ell$ are fixed constants in this SDP, and we optimize over $\lambda_1,\lambda_2,\ldots,\lambda_{mn}$:
	\begin{align}\begin{split}\label{sdp:eigs_imp_sfs_simp}
		\begin{matrix}
		\begin{tabular}{r l}
		\multicolumn{2}{c}{\text{Primal problem}} \\
		\text{minimize:} & $\mu_1\lambda_{mn} + {\rm min}\{\mu_1,1-\mu_1-\ell\}\lambda_{mn-1} + {\rm max}\{0,1-2\mu_1-\ell\}\lambda_{mn-2} + \ell\lambda_1$ \\
		\text{subject to:} & $\begin{bmatrix}
		2\lambda_{mn} & \lambda_{mn-1}-\lambda_1 \\ \lambda_{mn-1}-\lambda_1 & 2\lambda_{mn-2}
		\end{bmatrix} \geq 0$ \\
		\ & $\lambda_i \geq \lambda_{i+1} \geq 0 \quad \forall \, 1 \leq i \leq mn-1$ \\
		\ & $\displaystyle\sum_{i=1}^{mn} \lambda_i = 1$. \\				
		\end{tabular}
		\end{matrix}
	\end{split}\end{align}

The dual problem can be constructed using standard techniques of semidefinite programming as found in~\cite{Wat04Lec7}.

\begin{align}\begin{split}
\begin{matrix}
	\begin{tabular}{r l}
		\multicolumn{2}{c}{\text{Dual problem}} \\
		\text{maximize:} & $t$ \\
		\text{subject to:} & $t - 2b + y_1 = \ell$ \\
		\ & $t + 2b + y_{mn-1} - y_{mn-2} = {\rm min}\{\mu_1,1-\mu_1-\ell\}$ \\
		\ & $t + 2c + y_{mn-2} - y_{mn-3} = {\rm max}\{0,1-2\mu_1-\ell\}$ \\
		\ & $t + 2a - y_{mn-1} \leq \mu_1$ \\
		\ & $t + y_{i+1} - y_{i} = 0 \quad \forall \, 1 \leq i \leq mn-4$ \\
		\ & $y_{i} \geq 0 \quad \forall \, 1 \leq i \leq mn-1$ \\
		\ & $\begin{bmatrix}
		a & b \\ b & c
		\end{bmatrix} \geq 0$.				
		\end{tabular}
		\end{matrix}
	\end{split}\end{align}
	
	It thus suffices to find a feasible point of the above dual problem with $t = 0$. We note that code that implements the above SDP in MATLAB via the CVX package~\cite{cvx} can be downloaded from~\cite{AJR14code}. We now split into three cases, depending on which branch of $f$ we are working with.
	
	{\bf Case a): $-\tfrac{1}{2} \leq \ell \leq -\frac{1}{2\sqrt{2}}$}. In this case, we have $\mu_1 = (\sqrt{1-4\ell^2}-2\ell+1)/4$, ${\rm min}\{\mu_1,1-\mu_1-\ell\} = \mu_1$, and ${\rm max}\{0,1-2\mu_1-\ell\} = 1-2\mu_1-\ell$. It is then straightforward to verify the following defines a feasible point of the dual problem of the semidefinite program~\eqref{sdp:eigs_imp_sfs_simp}:
	\begin{align*}
		t & = 0, \quad \quad a = \frac{\ell+2\mu_1}{2}, \quad \quad b = -\frac{\ell}{2}, \quad \quad c = \frac{1-2\mu_1-\ell}{2} \\
		y_{i} & = 0 \quad \forall \, 1 \leq i \leq mn-2, \quad \quad y_{mn-1} = \mu_1+\ell.
	\end{align*}
	The only condition in the dual problem that is not obviously satisfied is the fact that $\begin{bmatrix}a & b \\ b & c \end{bmatrix} \geq 0$. However, this follows from the fact that $b^2 = ac$ for this particular choice of $a,b,c,$ and $\mu_1$. Since this dual feasible point has $t = 0$, it follows that the semidefinite program~\eqref{sdp:eigs_imp_sfs_simp} has optimal value $\geq 0$, as desired.
	
	{\bf Case b): $-\frac{1}{2\sqrt{2}} < \ell < \frac{1-\sqrt{2}}{2}$}. This case follows immediately from choosing $\ell = -\frac{1}{2\sqrt{2}}$ in case~a) and noting that we can choose the function $f$ described by the lemma to be non-decreasing.
	
	{\bf Case c): $\frac{1-\sqrt{2}}{2} \leq \ell \leq 0$}. In this case, we have $\mu_1 = (\sqrt{1+4\ell-4\ell^2}-2\ell+3)/4$, ${\rm min}\{\mu_1,1-\mu_1-\ell\} = 1-\mu_1-\ell$, and ${\rm max}\{0,1-2\mu_1-\ell\} = 0$. It is then straightforward to verify the following defines a feasible point of the dual problem of the semidefinite program~\eqref{sdp:eigs_imp_sfs_simp}:
	\begin{align*}
		t & = 0, \quad \quad a = \frac{\mu_1}{2}, \quad \quad b = \frac{1-\mu_1-\ell}{2}, \quad \quad c = \frac{1-\mu_1}{2} \\
		y_{i} & = 1-\mu_1 \quad \forall \, 1 \leq i \leq mn-3, \quad \quad y_{mn-2} = y_{mn-1} = 0.
	\end{align*}
	Similar to case~a), we have $b^2 = ac$ for this particular choice of $a,b,c,$ and $\mu_1$, so the above point indeed satisfies all of the constraints of the dual problem. Since $t = 0$, it follows that the semidefinite program~\eqref{sdp:eigs_imp_sfs_simp} has optimal value $\geq 0$, which completes the proof.
\end{proof}

\subsection{The Realignment Criterion}\label{sec:realignment}

The \emph{realignment criterion} \cite{CW03,R03} for entanglement states that all separable states $\rho \in M_m \otimes M_n$ satisfy $\|R(\rho)\|_{\textup{tr}} \leq 1$, where $R : M_m \otimes M_n \rightarrow M_{m,n} \otimes M_{m,n}$ is the linear ``realignment'' map defined on elementary tensors by $R(\ketbra{i}{j} \otimes \ketbra{k}{\ell}) = \ketbra{i}{k} \otimes \ketbra{j}{\ell}$. Thus if $\|R(\rho)\|_{\textup{tr}} > 1$ then we know that $\rho$ is entangled, and we say that the realignment criterion detected the entanglement in $\rho$. This criterion is particularly useful, as it is one of the simplest tests that can detect entanglement in PPT states. The main result of this section shows that the realignment criterion \emph{cannot} detect entanglement in any absolutely PPT states.

To phrase our result in another way, we can consider the sets of absolutely PPT states and ``absolutely realignable'' states:
\begin{align*}
	\mathcal{PPT}_{\textup{abs}} & \defeq \big\{ \rho : (id_m \otimes T)(U\rho U^\dagger) \geq 0 \quad \forall \, \text{unitary } U \big\}, \\
	\mathcal{R}_{\textup{abs}} & \defeq \big\{ \rho : \|R(U\rho U^\dagger)\|_{\textup{tr}} \leq 1 \quad \forall \, \text{unitary } U \big\}.
\end{align*}
Our result states that $\mathcal{PPT}_{\textup{abs}} \subseteq \mathcal{R}_{\textup{abs}}$, so the realignment criterion becomes a weaker entanglement test than the PPT criterion in the ``absolute'' setting:
\begin{theorem}\label{thm:realign_from_spec}
	If $\rho \in M_m \otimes M_n$ is absolutely PPT then $\|R(\rho)\|_{\textup{tr}} \leq 1$. That is, the realignment criterion cannot detect entanglement in any absolutely PPT state.
\end{theorem}
It will be helpful to note that for any state $\rho \in M_m \otimes M_n$, one can write $\rho$ in terms of its \emph{operator-Schmidt decomposition}, which is defined as
\begin{align*}
	\rho = \sum_i \lambda_i A_i \otimes B_i,
\end{align*} 
where $\lambda_i \geq 0$ for all $i$ and the sets of operators $\{A_i\}$ and $\{B_i\}$ form orthonormal bases of $M_m$ and $M_n$ in the Hilbert--Schmidt inner product. There is a well-known correspondence between the realignment criterion and the operator-Schmidt decomposition of any state $\rho$. Specifically, it is the case that $\|R(\rho)\|_{\textup{tr}} = \sum_i \lambda_i$ (a proof of this fact can be found in \cite{Joh12}), so any state with $\sum_i \lambda_i > 1$ is entangled. In this case, it is straightforward to see that the operator
\begin{align}\label{eq:realign_witness}
	W := I - \sum_i A_i \otimes B_i
\end{align}
is an entanglement witness that detects the entanglement in $\rho$, since orthonormality of $\{A_i\}$ and $\{B_i\}$ implies that ${\rm Tr}(W\rho) = 1 - \sum_i \lambda_i < 0$. Thus, to show that $\|R(\rho)\|_{\textup{tr}} \leq 1$ for all absolutely PPT states, it suffices to show that operators $W$ of the form~\eqref{eq:realign_witness} have $\Tr(W\rho) \geq 0$ whenever $\rho$ is absolutely PPT. In order to prove this result, we first need the following two auxiliary lemmas.
\begin{lemma}\label{lem:wit_trace_bound}
	Let $\{A_i\}_{i=1}^{m^2} \subset M_m$ and $\{B_i\}_{i=1}^{n^2} \subset M_n$ be orthonormal sets in the Hilbert--Schmidt inner product. Then $\big|{\rm Tr}\big(\sum_{i=1}^{k} A_i \otimes B_i\big)\big| \leq \sqrt{mn}$ for all $k \leq \min\{m^2,n^2\}$.
\end{lemma}
\begin{proof}
	Define vectors $u,v \in \mathbb{C}^k$ by $u_i := \Tr(A_i)$ and $v_i := \Tr(B_i)$. By the Cauchy--Schwarz inequality, it follows that
\begin{align*}
	\left|{\rm Tr}\left(\sum_{i=1}^{k} A_i \otimes B_i\right)\right| = \big|{\ip{\overline{u}}{v}}\big| \leq \|u\| \cdot \|v\|.
\end{align*}
To prove our result, we show that $\|{u}\|\leq \sqrt{m}$, with a similar argument holding for $\|{v}\| \leq \sqrt{n}$. We first rewrite $u_i$ as $u_i=\ip{I}{A_i}$. Note that the identity operator can be rewritten as
\begin{equation}
	\label{eq:identity-norm}
	\begin{aligned}
		I=\sum_{i=1}^{m^2}\ip{A_i}{I}A_i.
	\end{aligned}
\end{equation}

Computing the norms of the quantities in Equation~\eqref{eq:identity-norm}, we obtain
\begin{align*}
	\|{I}\|_F^2 = \sum_{i=1}^{m^2}|{\ip{A_i}{I}}|^2 \geq \sum_{i=1}^k|{u_i}|^2 = \|u\|^2,
\end{align*}
where the first equality follows from the definition of the Frobenius norm, and by noting that the set of operators $\{A_i\}_{i=1}^{m^2}$ are orthonormal. Since $\|{I}\|_F^2 = m$, it directly follows that $\|{u}\|\leq \sqrt{m}$, as desired.
\end{proof} 

\begin{lemma}\label{lem:realign_eig}
	Let $W \in M_m \otimes M_n$ be an entanglement witness of the form~\eqref{eq:realign_witness}, scaled so that ${\rm Tr}(W) = 1$. Then
	\begin{align*}
		\|W\|_{F} & \leq \sqrt{\frac{2}{mn - \sqrt{mn}}}.
	\end{align*}
\end{lemma}
\begin{proof}
	Consider the Hermitian operator $\widetilde{W} := I - \sum_{i=1}^k A_i \otimes B_i$ and let $W$ be its normalization: $W = \widetilde{W}/{\rm Tr}(\widetilde{W})$. For brevity, define $t := {\rm Tr}\big( \sum_{i=1}^k A_i \otimes B_i \big)$. Then
	\begin{align*}
		\|W\|_{F}^2 & = {\rm Tr}(W^\dagger W) = \frac{mn - 2t + k}{{\rm Tr}(\widetilde{W})^2} \leq \frac{2mn - 2t}{(mn - t)^2} = \frac{2}{mn - t} \leq \frac{2}{mn - \sqrt{mn}},
	\end{align*}
	where the first inequality above comes from simply noting that $k \leq \min\{m^2,n^2\} \leq mn$ and the second inequality comes from Lemma~\ref{lem:wit_trace_bound} and noting that $2 / (mn - t)$ is an increasing function of $t$ (for $t < mn$).
\end{proof}

We are now in a position to prove Theorem~\ref{thm:realign_from_spec}.
\begin{proof}[Proof of Theorem~\ref{thm:realign_from_spec}]
	Our goal is to make use of Lemma~\ref{lem:eigs_imp_sfs}, so we want to find bounds on $\ell := (1 - \|W\|_{\textup{tr}})/2$ and the maximum eigenvalue $\mu_1$ of $W$, where $W$ is an entanglement witness of the form~\eqref{eq:realign_witness}, scaled so that ${\rm Tr}(W) = 1$. Bounds on both of these quantities follow straightforwardly from Lemma~\ref{lem:realign_eig}. More specifically, it is always the case that $\|\cdot\|_{\textup{tr}} \leq \sqrt{d}\|\cdot\|_{F}$ for $d \times d$ matrices, so Lemma~\ref{lem:realign_eig} immediately implies that
	\begin{align*}
		\|W\|_{\textup{tr}} \leq \sqrt{mn}\|W\|_{F} \leq \sqrt{\frac{2\sqrt{mn}}{\sqrt{mn} - 1}}.
	\end{align*}
	Similarly, $\|\cdot\| \leq \|\cdot\|_{F}$ always, so we get the following bounds on $\ell$ and $\mu_1$:
		
	\begin{align*}
		\ell & \geq \frac{1}{2}\Big( 1 - \sqrt{\frac{2\sqrt{mn}}{\sqrt{mn} - 1}} \Big),\\
		\mu_1 & \leq \|W\| \leq \|W\|_F \leq \sqrt{\frac{2}{mn - \sqrt{mn}}}.
	\end{align*}
	Furthermore, when $m,n \geq 3$ (which is the only case we need to consider, since it is already known that absolute separability and absolute PPT coincide when $\min\{m,n\} \leq 2$) we then have the looser bounds $\ell \geq \frac{1}{2}(1 - \sqrt{3})$ and $\mu_1 \leq \frac{1}{\sqrt{3}}$. Since $f( \frac{1}{2}(1 - \sqrt{3})) \approx 0.6033\ldots \geq 0.5773\ldots \approx \frac{1}{\sqrt{3}} \geq \mu_1$, it follows from Lemma~\ref{lem:eigs_imp_sfs} that if $\rho$ is absolutely PPT then ${\rm Tr}(W\rho) \geq 0$, so $\|R(\rho)\|_{\textup{tr}} \leq 1$, as desired.
\end{proof}

\subsection{The Choi Map}\label{sec:choi_proof}

The \emph{Choi map} \cite{Cho80} is a positive map on $M_3$ that is defined as follows:
\begin{align*}
	\Phi_{C}(X) \defeq \frac{1}{2}\begin{bmatrix}
		x_{11} + x_{22} & -x_{12} & -x_{13} \\
		-x_{21} & x_{22} + x_{33} & -x_{23} \\
		-x_{31} & -x_{32} & x_{33} + x_{11}
	\end{bmatrix}.
\end{align*}
This map is one of the most well-known positive maps because it was one of the first maps found with the property that there are states $\rho \in M_3 \otimes M_3$ such that $(id_3 \otimes T)(\rho) \geq 0$, but $(id_3 \otimes \Phi_C)(\rho) \not\geq 0$. In other words, the Choi map was one of the first known examples of a positive map that can detect entanglement in PPT states. Furthermore, it is extremal in the set of positive maps \cite{Ha13}.

Our main result of this section is that the Choi map \emph{cannot} detect entanglement in absolutely PPT states. Equivalently, we show that the set of ``absolutely Choi map'' states:
\begin{align*}
	\mathcal{C}_{\textup{abs}} & \defeq \big\{ \rho : (id_3 \otimes \Phi_C)(U\rho U^\dagger) \geq 0 \quad \forall \, \text{unitary } U \big\}
\end{align*}
is a superset of the absolutely PPT states: $\mathcal{PPT}_{\textup{abs}} \subseteq \mathcal{C}_{\textup{abs}}$. Thus the Choi map is a weaker entanglement test than the PPT criterion in the ``absolute'' setting:
\begin{theorem}\label{thm:choi_from_spec}
	If $\rho \in M_3 \otimes M_3$ is absolutely PPT then $(id_3 \otimes \Phi_C)(\rho) \geq 0$.
\end{theorem}

As with the previous section, our first goal here is to rephrase the condition $(id_3 \otimes \Phi_C)(\rho) \not\geq 0$ in terms of entanglement witnesses, so that we can make use of Lemma~\ref{lem:eigs_imp_sfs}. To this end, simply note that if $(id_3 \otimes \Phi_C)(\rho) \not\geq 0$ then there exists a pure state $\ket{v} \in \mathbb{C}^3 \otimes \mathbb{C}^3$ such that $\bra{v}(id_3 \otimes \Phi_C)(\rho)\ket{v} < 0$. Thus $\Tr((id_3 \otimes \Phi_C^\dagger)(\ketbra{v}{v})\rho) < 0$, so
\begin{align}\label{eq:choi_wit}
	W := (id_3 \otimes \Phi_C^\dagger)(\ketbra{v}{v})
\end{align}
is a witness that detects entanglement in $\rho$. It thus suffices to show that witnesses of the form~\eqref{eq:choi_wit} cannot detect entanglement in absolutely PPT states $\rho$. In order to prove this claim, we present the following lemma, which bounds the eigenvalues of $(id_3 \otimes \Phi_{C}^\dagger)(\ketbra{v}{v})$.
\begin{lemma}\label{lem:choi_eigs}
	Let $\ket{v} \in \mathbb{C}^3 \otimes \mathbb{C}^3$ be a unit vector. Then the eigenvalues of $(id_3 \otimes \Phi_{C}^\dagger)(\ketbra{v}{v})$ are contained within the interval $[-1/6,2/3]$.
\end{lemma}
Before proving this result, we note that both of its bounds on the eigenvalues are tight. A minimal eigenvalue of $-1/6$ is obtained when $\ket{v} = \ket{\psi^+}$ (the standard maximally-entangled state), and a maximal eigenvalue of $2/3$ is obtained when $\ket{v} = \tfrac{1}{\sqrt{3}}\ket{0}\otimes(\ket{0}+\sqrt{2}\ket{1})$.
\begin{proof}
	Let $\mu_1 \geq \cdots \geq \mu_9$ be the eigenvalues of $(id_3 \otimes \Phi_{C}^\dagger)(\ketbra{v}{v})$. Our first goal is to show that $\mu_9 \geq -1/6$. To this end, first notice that $\Phi_{C}^\dagger$ is trace-preserving, so we have
	\begin{align}\label{eq:eig_sum_1}
		\sum_{i=1}^9 \mu_i = 1.
	\end{align}
	
	Also notice that
	\begin{align}\label{eq:eig_sum_2}
		\left(\sum_{i=1}^8 \mu_i\right) - \mu_9 \leq \sum_{i=1}^9 |\mu_i| \leq \big\| \Phi_{C}^\dagger \big\|_\diamond,
	\end{align}
	where $\|\cdot\|_\diamond$ is the \emph{diamond norm} \cite{Kit97} defined by
	\begin{align*}
		\big\|\Phi\big\|_\diamond \defeq \sup\Big\{ \big\|(id_3 \otimes \Phi)(X)\big\|_{\textup{tr}} : \big\|X\big\|_{\textup{tr}} \leq 1 \Big\}.
	\end{align*}
	By subtracting Inequality~\eqref{eq:eig_sum_2} from Equation~\eqref{eq:eig_sum_1}, we see that
	\begin{align}\label{eq:eig_sum_3}
		\mu_9 \geq \frac{1}{2}\left(1 - \big\| \Phi_{C}^\dagger \big\|_\diamond \right).
	\end{align}
	
	The diamond norm $\big\|\Phi_{C}^\dagger\big\|_\diamond$ can be computed via semidefinite programming \cite{Wa09}, and in particular is equal to the optimal value of the following problem \cite{Wat12}:
	\begin{align}\begin{split}\label{sdp:diamond_norm}
		\text{minimize:} \ \ & \frac{1}{2}\big\|{\rm Tr}_2(Y_0)\big\| + \frac{1}{2}\big\|{\rm Tr}_2(Y_1)\big\| \\
		\text{subject to:} \ \ & \begin{bmatrix}
		Y_0 & -J(\Phi_{C}^\dagger) \\ -J(\Phi_{C}^\dagger)^\dagger & Y_1
		\end{bmatrix} \geq 0 \\
		& Y_0,Y_1 \geq 0,
	\end{split}\end{align}
	where we optimize over $Y_0,Y_1 \in M_3 \otimes M_3$ and denote $\Tr_2(\cdot)$ as the partial trace with respect to the second subsystem. It is straightforward to verify that the following are feasible values of $Y_0$ and $Y_1$, written with respect to the standard basis of $\mathbb{C}^3 \otimes \mathbb{C}^3$:
	\begin{align*}
		Y_0=Y_1=\frac{1}{6}\begin{bmatrix}
		5 & \cdot & \cdot & \cdot & -1 & \cdot & \cdot & \cdot & -1 \\
		\cdot & 3 & \cdot & \cdot & \cdot & \cdot & \cdot & \cdot & \cdot \\
		\cdot & \cdot & \cdot & \cdot & \cdot & \cdot & \cdot & \cdot & \cdot \\
		\cdot & \cdot & \cdot & \cdot & \cdot & \cdot & \cdot & \cdot & \cdot \\
		-1 & \cdot & \cdot & \cdot & 5 & \cdot & \cdot & \cdot & -1 \\
		\cdot & \cdot & \cdot & \cdot & \cdot & 3 & \cdot & \cdot & \cdot \\
		\cdot & \cdot & \cdot & \cdot & \cdot & \cdot & 3 & \cdot & \cdot \\
		\cdot & \cdot & \cdot & \cdot & \cdot & \cdot & \cdot & \cdot & \cdot \\
		-1 & \cdot & \cdot & \cdot & -1 & \cdot & \cdot & \cdot & 5
		\end{bmatrix}.
	\end{align*}
	Since $\frac{1}{2}\big\|{\rm Tr}_2(Y_0)\big\| + \frac{1}{2}\big\|{\rm Tr}_2(Y_1)\big\| = 4/3$, it follows that $\big\|\Phi_{C}\big\|_\diamond \leq 4/3$ (it is not difficult to show that it actually equals $4/3$, but we only need the upper bound). By plugging this upper bound into Inequality~\eqref{eq:eig_sum_3}, we see that $\mu_9 \geq -1/6$, as desired.
	
	Our next goal is to show that $\mu_1 \leq 2/3$. To show this, we first note that
	\begin{align*}
		\mu_1 \leq \sup_{\ket{v} \in \mathbb{C}^3 \otimes \mathbb{C}^3} \Big\{ \big\| (id_3 \otimes \Phi_{C}^\dagger)(\ketbra{v}{v}) \big\| \Big\}.
	\end{align*}
	By making use of \cite[Theorem~4]{Wat05}, we see that the $id_3$ in this quantity can be omitted, giving
	\begin{align}\label{eq:mu1_ub2}
		\mu_1 \leq \sup_{\ket{v} \in \mathbb{C}^3} \Big\{ \big\| \Phi_{C}^\dagger(\ketbra{v}{v}) \big\| \Big\}.
	\end{align}
	It was shown in \cite{NOP09} that the quantity on the right in Inequality~\eqref{eq:mu1_ub2} equals
	\begin{align}\label{eq:inf_norm_sep}
		\sup_{\ket{v},\ket{w} \in \mathbb{C}^3} \big\{ (\bra{v}\otimes\bra{w})J(\Phi_{C}^\dagger)(\ket{v}\otimes\ket{w}) \big\}.
	\end{align}
	Optimizing over separable states in general is expected to be difficult, but we can compute upper bounds of the quantity~\eqref{eq:inf_norm_sep} via the semidefinite programming methods of \cite{NOP09,JK11}. In particular, the following SDP takes the supremum over the set of PPT states, rather than the set of separable states, and thus computes an upper bound of the quantity~\eqref{eq:inf_norm_sep} (and hence of $\mu_1$):
	\begin{align}\begin{split}\label{sdp:eig_ub}
		\begin{matrix}
		\begin{tabular}{r l c r l}
		\multicolumn{2}{c}{\text{Primal problem}} & & \multicolumn{2}{c}{\text{Dual problem}} \\
		\text{maximize:} & $\Tr(J(\Phi_{C}^\dagger)\rho)$ & &  \text{minimize:} & $\lambda_{\textup{max}}\big((id_3 \otimes T)(Y) + J(\Phi_{C}^\dagger)\big)$ \\
		\text{subject to:} & $(id_3 \otimes T)(\rho) \geq 0$ & \quad \quad & \text{subject to:} & $Y \geq 0$, \\
		\ & ${\rm Tr}(\rho) \leq 1$ & & & \\
		\ & $\rho \geq 0$ & & &\\
		\end{tabular}
		\end{matrix}
	\end{split}
	\end{align}
	where we optimize over density matrices $\rho \in M_3 \otimes M_3$ in the primal problem and over Hermitian $Y \in M_3 \otimes M_3$ in the dual problem. It is straightforward to use semidefinite programming solvers to numerically verify that the optimal value of this semidefinite program is $2/3$, from which it follows that $\mu_1 \leq 2/3$. To obtain a completely analytic proof of this fact, it suffices to find a single positive semidefinite $Y \in M_3 \otimes M_3$ such that $\lambda_{\textup{max}}\big((id \otimes T)(Y) + J(\Phi_{C}^\dagger)\big) = 2/3$. One such matrix is as follows:
	\begin{align*}
		Y = \frac{1}{6}\begin{bmatrix}
		\cdot & \cdot & \cdot & \cdot & \cdot & \cdot & \cdot & \cdot & \cdot \\
		 \cdot & 1 & \cdot & 2 & \cdot & \cdot & \cdot & \cdot & \cdot \\
		 \cdot & \cdot & 4 & \cdot & \cdot & \cdot & 2 & \cdot & \cdot \\
		 \cdot & 2 & \cdot & 4 & \cdot & \cdot & \cdot & \cdot & \cdot \\
		 \cdot & \cdot & \cdot & \cdot & \cdot & \cdot & \cdot & \cdot & \cdot \\
		 \cdot & \cdot & \cdot & \cdot & \cdot & 1 & \cdot & 2 & \cdot \\
		 \cdot & \cdot & 2 & \cdot & \cdot & \cdot & 1 & \cdot & \cdot \\
		 \cdot & \cdot & \cdot & \cdot & \cdot & 2 & \cdot & 4 & \cdot \\
		 \cdot & \cdot & \cdot & \cdot & \cdot & \cdot & \cdot & \cdot & \cdot
		\end{bmatrix}.
	\end{align*}
\end{proof}

With the above lemma in hand, we are now in a position to prove the main result of this section.
\begin{proof}[Proof of Theorem~\ref{thm:choi_from_spec}]
	By Lemma~\ref{lem:choi_eigs}, we know that every entanglement witness $W$ of the form $W = (id_3 \otimes \Phi_{C}^\dagger)(\ketbra{v}{v})$ has eigenvalues in the interval $[-1/6,2/3]$. Since $\Phi_C$ can be written in the form $\Phi_C = (\Phi - id_3)/2$ for some completely positive map $\Phi$, we see that $W = ((id_3 \otimes \Phi^\dagger)(\ketbra{v}{v}) - \ketbra{v}{v})/2$, which has at most $1$ negative eigenvalue, so we know that $\ell := (1 - \|W\|_{\textup{tr}})/2 \geq -1/6$ as well. Since $f(-1/6) = \frac{1}{12}(10 + \sqrt{2}) \approx 0.9512\ldots \geq \frac{2}{3}$, it follows from Lemma~\ref{lem:eigs_imp_sfs} that $W$ cannot detect entanglement in any absolutely PPT state, so neither can the Choi map $\Phi_C$.
\end{proof}

\subsection{Generalized Choi Maps}\label{sec:choi_generalize}

In the proof of Theorem~\ref{thm:choi_from_spec}, there was a rather large gap between the largest eigenvalue of $(id_3 \otimes \Phi_{C}^\dagger)(\ketbra{v}{v})$ (which was $2/3$) and the quantity that we needed to bound this eigenvalue by (which was $\frac{1}{12}(10 + \sqrt{2}) \approx 0.9512\ldots$). This suggests that any positive map that is sufficiently close to the Choi map also cannot detect entanglement in absolutely PPT states, since perturbing the Choi map will only slightly change both of these values. This section is devoted to making this statement more rigorous and precise, by investigating a well-known family of positive maps that generalize the Choi map.

We now introduce an infinite family of positive maps based on two real parameters $b,c\geq 0$ that were first studied in~\cite{CKL92} (see also \cite{CW11,CS13}). If we let $a := 2-b-c$, then these maps are defined as follows:
\begin{align*}
	\Phi_{b,c}(X) \defeq \frac{1}{2}\begin{bmatrix}
		ax_{11} + bx_{22} + cx_{33} & -x_{12} & -x_{13} \\
		-x_{21} & cx_{11} + ax_{22} + bx_{33} & -x_{23} \\
		-x_{31} & -x_{32} & bx_{11} + cx_{22} + ax_{33}
	\end{bmatrix}.
\end{align*}
Notice that the Choi map $\Phi_C$ is recovered in the $(b,c) = (1,0)$ case. Additionally, in the $(b,c) = (1,1)$ case the map has the form $\Phi_{1,1}(X) = \frac{1}{2}\big(\Tr(X)I - X\big)$, which is the well-known \emph{reduction map} \cite{HH99}.

It is known that $\Phi_{b,c}$ is positive but not completely positive (i.e., capable of detecting entanglement in some state) if and only if $(b,c) \neq (0,0)$, and either $b+c \leq 1$ or $bc \geq (b+c-1)^2$ (or both). Furthermore, it is \emph{indecomposable} (i.e., it detects some PPT entanglement) if and only if $b \neq c$, and it is an exposed point in the convex set of positive maps (and hence extreme and optimal) \cite{HK11} if $b \neq c$, $b+c > 1$, and $bc = (b+c-1)^2$ (see Figure~\ref{fig:phi_bc}).
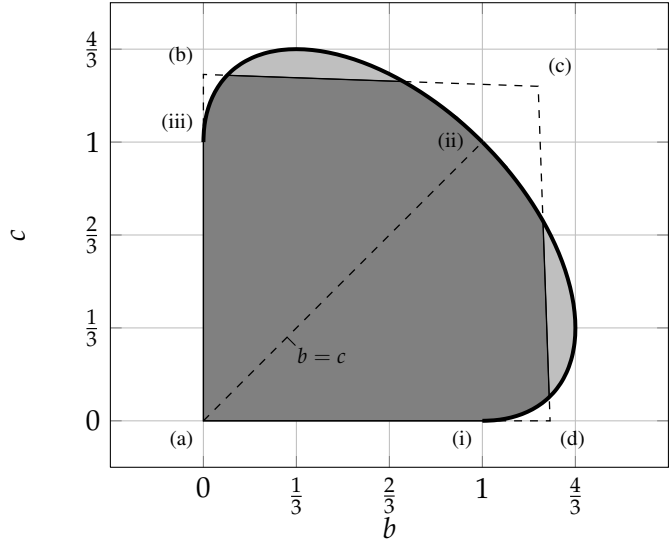
\begin{figure}[htb]
	\centering
\begin{tikzpicture}
\begin{axis}[
xlabel=$b$,
ylabel=$c$,
axis equal,
xmin=0, xmax=1.3333333333333333,
ymin=-0.16666666666, ymax=1.5,
xtick={0,0.3333333333333333,0.6666666666667,1,1.3333333333333333},
ytick={0,0.3333333333333333,0.6666666666667,1,1.3333333333333333},
xticklabels={$0$,$\tfrac{1}{3}$,$\tfrac{2}{3}$,$1$,$\tfrac{4}{3}$},
yticklabels={$0$,$\tfrac{1}{3}$,$\tfrac{2}{3}$,$1$,$\tfrac{4}{3}$},
grid=major,
]

\draw[domain=0:1,smooth,variable=\t,fill=black!25,line width=0.001] plot ({100*\t*\t/(1-\t+\t*\t)},{100*(0.16666666666+1/(1-\t+\t*\t))});
\draw[domain=0:1,smooth,variable=\t,fill=black!25,line width=0.001] plot ({100/(1-\t+\t*\t)},{100*(0.16666666666+\t*\t/(1-\t+\t*\t))});

\addplot[fill=black!50,line width=0.001] coordinates {(1.09009,0.909001) (1,1) (0.909001,1.09009) (0.734495,1.2071) (0.717075,1.21716) (0.0850174,1.23962) (0.0493374187,1.193303268) (0.0343145,1.1656855) (0.0199099,1.1300901) (0.00916698,1.09083302) (0.00238637,1.04761363) (0,1) (0,0) (1,0) (1.04761363,0.00238637) (1.09083302,0.00916698) (1.1300901,0.0199099) (1.1656855,0.0343145) (1.193303268,0.0493374187) (1.23962,0.0850174) (1.21716,0.717075) (1.2071,0.734495)};

\draw[domain=0:1,smooth,variable=\t,black,line width=1.5] plot ({100*\t*\t/(1-\t+\t*\t)},{100*(0.16666666666+1/(1-\t+\t*\t))});
\draw[domain=0:1,smooth,variable=\t,black,line width=1.5] plot ({100/(1-\t+\t*\t)},{100*(0.16666666666+\t*\t/(1-\t+\t*\t))});

\addplot[black,line width=0.5] coordinates {(1,0) (0,0) (0,1)};
\node at (axis cs:0.3,0.3) [anchor=north west] {{\scriptsize $b=c$}};
\addplot[black] coordinates {(0.3,0.3) (0.33,0.27)};
\addplot[black,line width=0.5,dashed] coordinates {(0,0) (1,1)};

\addplot[black,line width=0.5,dashed] coordinates {(0,1) (0,1.242640687) (0.0850174,1.23962)};
\addplot[black,line width=0.5] coordinates {(0.0850174,1.23962) (0.717075,1.21716)};
\addplot[black,line width=0.5,dashed] coordinates {(0.717075,1.21716) (1.2,1.2) (1.21716,0.717075)};
\addplot[black,line width=0.5] coordinates {(1.23962,0.0850174) (1.21716,0.717075)};
\addplot[black,line width=0.5,dashed] coordinates {(1,0) (1.242640687,0) (1.23962,0.0850174)};

\addplot[black,fill=white] (1,0) circle [radius=0.015];
\node at (axis cs:1,0) [anchor=north east] {{\scriptsize (i)}};

\addplot[black,fill=white] (1,1) circle [radius=0.015];
\node at (axis cs:0.97,1) [anchor=east] {{\scriptsize (ii)}};

\addplot[black,fill=white] (0,1) circle [radius=0.015];
\node at (axis cs:0,1) [anchor=south east] {{\scriptsize (iii)}};

\addplot[black,fill=black] (0,0) circle [radius=0.01];
\node at (axis cs:0,0) [anchor=north east] {{\scriptsize (a)}};

\addplot[black,fill=black] (0,1.242640687) circle [radius=0.01];
\node at (axis cs:0,1.242640687) [anchor=south east] {{\scriptsize (b)}};

\addplot[black,fill=black] (1.2,1.2) circle [radius=0.01];
\node at (axis cs:1.2,1.2) [anchor=south west] {{\scriptsize (c)}};

\addplot[black,fill=black] (1.242640687,0) circle [radius=0.01];
\node at (axis cs:1.242640687,0) [anchor=north west] {{\scriptsize (d)}};
\end{axis}
\end{tikzpicture}
	\caption{A plot of the set of positive but not completely positive maps $\Phi_{b,c}$. The gray region indicates the values of $(b,c)$ for which $\Phi_{b,c}$ is positive. For every point $(b,c)$ in the gray region, with the exception of the dashed line $b=c$, the map $\Phi_{b,c}$ is indecomposible, and hence is able to detect entanglement in some PPT state. The curved part of the boundary (i.e., the thick black line) consists of exposed positive maps. The four points $(a)$--$(d)$ are the points described by Theorem~\ref{thm:gen_choi_from_spec}, and hence every map in the dark gray region is incapable of detecting entanglement in absolutely PPT states. The point~(i) is the Choi map $\Phi_C$, (iii) is its dual $\Phi_C^\dagger$, and (ii) is the reduction map. This figure is also reproduced and considered in the following works~\cite{CS13, CW11}}\label{fig:phi_bc}
\end{figure}

As with the previous sections, our goal is to show that the maps $\Phi_{b,c}$ cannot detect entanglement in any absolutely PPT state. While we are not able to prove this for \emph{all} positive $\Phi_{b,c}$, we are able to prove it for most of them, including many of which are exposed in the set of positive maps.
\begin{theorem}\label{thm:gen_choi_from_spec}
	Suppose that the point $(b,c)$ is contained within the convex hull of the following $4$ points:
	\begin{enumerate}[(a)]
		\item $(0, 0)$,
		\item $(0, 3(\sqrt{2}-1))$,
		\item $(6/5, 6/5)$,
		\item $(3(\sqrt{2}-1), 0)$.
	\end{enumerate}
	If $\rho \in M_3 \otimes M_3$ is absolutely PPT then $(id_3 \otimes \Phi_{b,c})(\rho) \geq 0$.
\end{theorem}
As with Theorem~\ref{thm:choi_from_spec}, our method of proof is to come up with bounds on the eigenvalues of $(id_3 \otimes \Phi_{b,c}^\dagger)(\ketbra{v}{v})$. By convexity arguments, it would suffice for our purposes to prove these bounds only for the four points described by Theorem~\ref{thm:gen_choi_from_spec}, but we prove bounds that hold for all $(b,c)$ since it is not much more work.
\begin{lemma}\label{lem:choi_eigs_gen}
	Let $\ket{v} \in \mathbb{C}^3 \otimes \mathbb{C}^3$ and suppose $b,c \geq 0$ satisfy $b+c \leq 3$. Then the eigenvalues of $(id_3 \otimes \Phi_{b,c}^\dagger)(\ketbra{v}{v})$ are no smaller than $-\frac{1}{6}(b+c)$.
\end{lemma}
\begin{proof}
	We note that the nuts and bolts of the proof are almost identical to the proof of the lower bound given in Lemma~\ref{lem:choi_eigs}, so it suffices to just present a solution to the semidefinite program~\eqref{sdp:diamond_norm} that works for all $\Phi_{b,c}$ rather than just $\Phi_C$. Indeed, it is straightforward to check that the following is a feasible point of the semidefinite program:
	\begin{align*}
		Y_0=Y_1=\frac{1}{6}\begin{bmatrix}
		6-b-c & \cdot & \cdot & \cdot & 2b+2c-3 & \cdot & \cdot & \cdot & 2b+2c-3 \\
		\cdot & 3b & \cdot & \cdot & \cdot & \cdot & \cdot & \cdot & \cdot \\
		\cdot & \cdot & 3c & \cdot & \cdot & \cdot & \cdot & \cdot & \cdot \\
		\cdot & \cdot & \cdot & 3c & \cdot & \cdot & \cdot & \cdot & \cdot \\
		2b+2c-3 & \cdot & \cdot & \cdot & 6-b-c & \cdot & \cdot & \cdot & 2b+2c-3 \\
		\cdot & \cdot & \cdot & \cdot & \cdot & 3b & \cdot & \cdot & \cdot \\
		\cdot & \cdot & \cdot & \cdot & \cdot & \cdot & 3b & \cdot & \cdot \\
		\cdot & \cdot & \cdot & \cdot & \cdot & \cdot & \cdot & 3c & \cdot \\
		2b+2c-3 & \cdot & \cdot & \cdot & 2b+2c-3 & \cdot & \cdot & \cdot & 6-b-c
		\end{bmatrix},
	\end{align*}
	where we have used the constraints on $b$ and $c$ given in the statement of the lemma to ensure that the positive semidefiniteness requirements of the SDP~\eqref{sdp:diamond_norm} are satisfied. The corresponding value of the objective function is
	\begin{align*}
		\big\|{\rm Tr}_2(Y_0)\big\| & = \frac{1}{3} \left\| \begin{bmatrix} 3+b+c & 0 & 0 \\ 0 & 3+b+c & 0 \\ 0 & 0 & 3+b+c\end{bmatrix} \right\| \\
		& = \frac{1}{3}(3+b+c).
	\end{align*}
	It follows that $\|\Phi_{b,c}^\dagger\|_\diamond \leq \frac{1}{3}(3+b+c)$, so from Equation~\eqref{eq:eig_sum_3} it follows that if $\mu_9$ is the minimal eigenvalue of $(id_3 \otimes \Phi_{b,c}^\dagger)(\ketbra{v}{v})$ then $\mu_9 \geq \frac{1}{2}(1 - \|\Phi_{b,c}^\dagger\|_\diamond) \geq -\frac{1}{6}(b+c)$, as desired.
\end{proof}

Upper bounds on the eigenvalues of $(id_3 \otimes \Phi_{b,c}^\dagger)(\ketbra{v}{v})$ seem to be much messier than the lower bound given by Lemma~\ref{lem:choi_eigs_gen}, so we defer their discussion to Appendix~A, where we complete the proof of Theorem~\ref{thm:gen_choi_from_spec}.

Theorem~\ref{thm:gen_choi_from_spec} shows analytically that all of the maps in the dark shaded region of Figure~\ref{fig:phi_bc} are unable to detect entanglement in absolutely PPT states. It is natural to ask whether or not the same is true of the positive maps in the light shaded region. We do not have an analytic proof that this is the case, but numerical evidence suggests that it is. In particular, we randomly generated $10^8$ pure states $\ket{v}$ and values of $(b,c)$ in the light shaded region of Figure~\ref{fig:phi_bc}, and found that every time the eigenvalues of $(id_3 \otimes \Phi_{b,c}^\dagger)(\ketbra{v}{v})$ satisfied the hypotheses of Lemma~\ref{lem:eigs_imp_sfs}.

\subsection{Breuer--Hall Map}\label{sec:breuer_hall}

The \emph{Breuer--Hall map} \cite{Bre06,Hal06} is a positive map on $M_n$, where $n \geq 4$ is even, that is defined as follows:
\begin{align*}
	\Phi_{BH}(X) \defeq \frac{1}{n-2}\left(\Tr(X)I - X - V X^{T} V^{\dagger}\right),
\end{align*}
where $V \in M_{n}$ is a given unitary matrix that is skew symmetric (anti-symmetric), i.e. $V^{T} = -V$. We note that the reason for restricting to even $n$ is because such unitary matrices exist if and only if $n$ is even. One such unitary matrix is as follows:
\begin{align*}
	V = \begin{bmatrix}
		0 & 0 & \cdots & 0 & 1 \\
		0 & 0 & \cdots & 1 & 0 \\
		\vdots & \vdots & \rddots & \vdots & \vdots \\
		0 & -1 & \cdots & 0 & 0 \\
		-1 & 0 & \cdots & 0 & 0		
	\end{bmatrix},
\end{align*}
and it follows straightforwardly from \cite[Theorem~2.3]{MMX10} that for our purposes it suffices to restrict attention to this particular $V$ (i.e., there exists a Breuer--Hall map detecting entanglement in some absolutely PPT state if and only if the Breuer--Hall map that arises from this particular $V$ detects entanglement in some absolutely PPT state).

Much like the Choi map and generalized Choi maps, the Breuer--Hall map is capable of detecting entanglement in PPT states. We now show, in a manner similar to the proof of Theorem~\ref{thm:choi_from_spec}, that these maps \emph{cannot} detect entanglement in absolutely PPT states. Phrased differently, we show that the set of ``absolutely Breuer--Hall'' states:
\begin{align*}
	\mathcal{BH}_{\textup{abs}} & \defeq \big\{ \rho : (id_{n} \otimes \Phi_{BH})(U\rho U^\dagger) \geq 0 \quad \forall \, \text{unitary } U \big\}
\end{align*}
is a superset of the absolutely PPT states: $\mathcal{PPT}_{\textup{abs}} \subseteq \mathcal{BH}_{\textup{abs}}$.
\begin{theorem}\label{thm:breuerhall_from_spec}
	Let $n \geq 4$ be even. If $\rho \in M_{n} \otimes M_{n}$ is absolutely PPT then $(id_{n} \otimes \Phi_{BH})(\rho) \geq 0$.
\end{theorem}
Just like in the proof of Theorem~\ref{thm:choi_from_spec}, our goal is to bound the eigenvalues of operators of the form $(id_n \otimes \Phi^{\dagger}_{BH})(\ketbra{v}{v})$. The following lemma provides such a bound.
\begin{lemma}\label{lem:breuerhall_eigs}
	Let $\ket{v} \in \mathbb{C}^n \otimes \mathbb{C}^n$, where $n \geq 4$ is even. Then the eigenvalues of $(id_n \otimes \Phi^{\dagger}_{BH})(\ketbra{v}{v})$ are contained within the interval $[-1/n, 1/(n-2)]$. 
\end{lemma}
\begin{proof}
The proof follows the same construction as the proof of Theorem~\ref{thm:choi_from_spec} and uses the same semidefinite programs. Let $\mu_1 \geq \cdots \geq \mu_{n^2}$ be the eigenvalues of $(id_{n} \otimes \Phi_{BH}^{\dagger})(\ket{v}\bra{v})$. We first show that $\mu_1 \leq 1/(n-2)$. To achieve this, we must find a feasible point $Y$ of the SDP~\eqref{sdp:eig_ub} (replacing $\Phi_C$ by $\Phi_{BH}$) such that the objective function has the value $1/(n-2)$ at that point. One choice of $Y$ that works is $Y = \frac{n}{n-2}(I \otimes V)\ketbra{\psi^+}{\psi^+}(I \otimes V^\dagger)$, which is clearly positive semidefinite and thus a feasible point. The corresponding value of the objective function is
\begin{align*}
	\lambda_{\textup{max}}\big( J(\Phi_{BH}^\dagger) + (id_n \otimes T)(Y) \big) = \frac{1}{n-2}\lambda_{\textup{max}}\big( I - n\ketbra{\psi^+}{\psi^+} \big) = \frac{1}{n-2},
\end{align*}
as desired.

We next show that $\mu_{n^2} \geq -1/n$. To proceed, we calculate the diamond norm $\|\Phi^{\dagger}_{BH}\|_{\diamond}$ using the SDP~\eqref{sdp:diamond_norm}, where we optimize over $Y_0, Y_1 \in M_n \otimes M_n$. We now show that the following choice of $Y_0$ and $Y_1$ is a dual feasible point that achieves this bound: 
\begin{align*}
	Y_0 = Y_1 & = J(\Phi_{BH}^\dagger) + 2\ketbra{\psi^+}{\psi^+}.
\end{align*}
	It is straightforward to show that $P := Y_0 - \ketbra{\psi^+}{\psi^+}$ is (up to scaling) an orthogonal projection, so $Y_0$ is positive semidefinite. Next observe that
	\begin{align*}
	\begin{bmatrix}
		Y_0 & -J(\Phi_{BH}^\dagger) \\ -J(\Phi_{BH}^\dagger)^\dagger & Y_1
	\end{bmatrix} = \begin{bmatrix}
		P & -P \\ -P & P
	\end{bmatrix} + \begin{bmatrix}
		\ketbra{\psi^+}{\psi^+} & \ketbra{\psi^+}{\psi^+} \\ \ketbra{\psi^+}{\psi^+} & \ketbra{\psi^+}{\psi^+}
	\end{bmatrix} \geq 0,
	\end{align*}
	from which it follows that the given choice of $Y_0,Y_1$ is a feasible point of the SDP~\eqref{sdp:diamond_norm}.
	
	To compute the corresponding value of the objective function, we note that ${\rm Tr}_2(Y_0) = \frac{n}{n-2}I - \frac{1}{n-2}I - \frac{1}{n-2}I + \frac{2}{n}I = \frac{n+2}{n}I$, so $\frac{1}{2}\big\|{\rm Tr}_2(Y_0)\big\| + \frac{1}{2}\big\|{\rm Tr}_2(Y_1)\big\| = \frac{n+2}{n}$. It follows that $\big\|\Phi_{BH}^\dagger\big\|_\diamond \leq \frac{n+2}{n}$. By plugging this upper bound into Inequality~\eqref{eq:eig_sum_3} we see that $\mu_{n^2} \geq -1/n$, as desired, which concludes the proof.
\end{proof}

Now that we have bounds on the eigenvalues of the witnesses of the form $(id_n \otimes \Phi^{\dagger}_{BH})(\ketbra{v}{v})$, we can prove the main result of this section.
\begin{proof}[Proof of Theorem~\ref{thm:breuerhall_from_spec}]
	By Lemma~\ref{lem:breuerhall_eigs}, we know that every entanglement witness $W$ of the form $W = (id_n \otimes \Phi_{BH}^\dagger)(\ketbra{v}{v})$ has eigenvalues in the interval $[-1/n, 1/(n-2)]$. From the definition of $\Phi_{BH}$ it is straightforward to see that every such $W$ has at most $1$ negative eigenvalue, so we know that $\ell := (1 - \|W\|_{\textup{tr}})/2 \geq -1/n$ as well.
	
	For all $n \geq 4$, we have $f(\ell) = f(-1/n) \geq f(-1/4) = (1+\sqrt{2})/4 \approx 0.6035\ldots \geq 1/2 \geq 1/(n-2) = \mu_1$. It follows from Lemma~\ref{lem:eigs_imp_sfs} that $W$ cannot detect entanglement in any absolutely PPT state, so neither can $\Phi_{BH}$.
\end{proof}

\section{Special Classes of States}\label{sec:special_states}

Up until now we have focused on proving that certain separability criteria are incapable of detecting entanglement in absolutely PPT states. Now we shift focus a bit and prove that the absolute separability and absolute PPT properties coincide on certain important families of quantum states.

\subsection{Werner States}\label{sec:werner}

The first family of states that we investigate are the \emph{Werner states} \cite{W89}, which are defined via the single real parameter $\alpha \in [-1,1]$ by
\begin{align*}
	\rho = (I - \alpha S)/(n^2 - n\alpha) \in M_n \otimes M_n,
\end{align*}
where $S \in M_n \otimes M_n$ is the \emph{swap} operator defined by $S(\ket{a}\otimes\ket{b}) = \ket{b}\otimes\ket{a}$. It is well-known that the Werner state $\rho$ is separable if and only if it is PPT if and only if $\alpha \leq 1/n$. We now investigate when these states are absolutely separable and absolutely PPT.
\begin{theorem}\label{thm:werner}
	Let $\rho = (I - \alpha S)/(n^2 - n\alpha) \in M_n \otimes M_n$ be a Werner state. Then $\rho$ is absolutely separable if $-1/n \leq \alpha \leq 1/n$, and it is not absolutely PPT if $\alpha < -1/(n-1)$ or $\alpha > 1/n$.
\end{theorem}
\begin{proof}
	Throughout this proof, we work with the operator $X := I - \alpha S$, which is equal to $\rho$ up to normalization (and thus has the same separability and PPT properties as $\rho$). We first prove that $|\alpha| \leq 1/n$ implies that $X$ is absolutely separable. To this end, notice that $UXU^\dagger = I - \alpha USU^\dagger$. Since $\|\alpha USU^\dagger\|_2 = \|\alpha S\|_2 = |\alpha|\|S\|_2 = |\alpha| n$, it follows from \cite[Theorem~1]{GB02} that $UXU^\dagger$ is separable whenever $|\alpha| n \leq 1$ (i.e., $|\alpha| \leq 1/n$).
	
	To prove the claim about when the given state is not absolutely PPT, we split into two separate cases.
	
	\textbf{Case 1: $\alpha > 1/n$}. Our goal is to show that $X$ is not absolutely PPT. Notice that the eigenvalues of $X$ are $1 - \alpha$ with multiplicity $n(n+1)/2$ and $1+\alpha$ with multiplicity $n(n-1)/2$. It follows from a straightforward calculation that each of the LMIs described by \cite{Hil07} corresponding to the set of absolutely PPT states is of the form
	\begin{align*}
		\begin{bmatrix}2-2\alpha & -2\alpha & -2\alpha & \cdots & -2\alpha \\
		-2\alpha & 2-2\alpha & -2\alpha & \cdots & -2\alpha \\
		-2\alpha & -2\alpha & 2-2\alpha & \cdots & -2\alpha \\
		\vdots & \vdots & \vdots & \ddots & \vdots \\
		-2\alpha & -2\alpha & -2\alpha & \cdots & 2-2\alpha\end{bmatrix} \geq 0.
	\end{align*}
	It is straightforward to see that the minimum eigenvalue of the above matrix is $2 - 2n\alpha$, so the matrix is positive semidefinite if and only if $2-2n\alpha \geq 0$, which does not hold in this case.
	
	\textbf{Case 2: $\alpha < -1/(n-1)$}. Again, our goal is to show that $X$ is not absolutely PPT. As in case~1 above, the eigenvalues of $X$ are $1 - \alpha$ with multiplicity $n(n+1)/2$ and $1+\alpha$ with multiplicity $n(n-1)/2$. However, this time there are many different LMIs corresponding to the absolutely PPT condition, since it is now the \emph{larger} eigenvalue that has higher multiplicity. One of the LMIs is
	\begin{align*}
		L_1 := \begin{bmatrix}2+2\alpha & 2\alpha & 2\alpha & \cdots & 2\alpha & 0 \\
		2\alpha & 2+2\alpha & 2\alpha & \cdots & 2\alpha & 0 \\
		2\alpha & 2\alpha & 2+2\alpha & \cdots & 2\alpha & 0 \\
		\vdots & \vdots & \vdots & \ddots & \vdots & \vdots \\
		2\alpha & 2\alpha & 2\alpha & \cdots & 2+2\alpha & 0 \\
		0 & 0 & 0 & \cdots & 0 & 2-2\alpha\end{bmatrix} \geq 0.
	\end{align*}
	Similar to in case~1, the minimum eigenvalue of the above matrix is $2+2(n-1)\alpha$, so the positive semidefiniteness constraint is satisfied if and only if $\alpha \geq -1/(n-1)$. In particular, this one LMI shows that $X$ is not absolutely PPT when $\alpha < -1/(n-1)$.
\end{proof}

Theorem~\ref{thm:werner} leaves open the question of whether or not Werner states are absolutely separable (or even absolutely PPT) when $-1/(n-1) \leq \alpha < -1/n$. We have used Hildebrand's characterization to numerically show that these states are all indeed absolutely PPT for all $n \leq 140$ (we note that even though there are an astronomical number of constraints in Hildebrand's characterization when $n = 140$, the vast majority of them are identical since so many of the eigenvalues of these states are equal). Furthermore, when $n = 3$ and $\alpha = -1/(n-1)$, we have numerically tested $10^4$ states of the form $U\rho U^\dagger$, and have not been able to find entanglement in any of them, suggesting that these states are absolutely separable if and only if they are absolutely PPT if and only if $\alpha \in [-1/(n-1), 1/n]$.

\subsection{Isotropic States}\label{sec:isotropic}

The family of isotropic states is defined by the single real parameter $\alpha \in \left[ -\frac{1}{n^2-1}, 1 \right]$ via
\begin{align}
\label{eq:isotropic-state}
\rho = \frac{1-\alpha}{n^2}I + \alpha \ket{\psi^+}\bra{\psi^+} \in M_n \otimes M_n,
\end{align}
where we recall that $\ket{\psi^+} = \frac{1}{\sqrt{n}} \sum_{i=1}^n \ket{i}\otimes\ket{i}$ is the standard maximally-entangled state. It is well-known that the isotropic state $\rho$ is separable if and only if it is PPT if and only if $\alpha \leq \frac{1}{n+1}$. We now show that the isotropic states that are absolutely separable and those that are absolutely PPT also coincide. 

\begin{theorem}
\label{thm:isotropic-thm}
Let $\rho = \frac{1-\alpha}{n^2}I + \alpha \ketbra{\psi^+}{\psi^+} \in M_n \otimes M_n$ be an isotropic state. Then $\rho$ is absolutely separable if and only if $\rho$ is absolutely PPT if and only if $\alpha \leq \frac{2}{2+n^2}$. \end{theorem}

\begin{proof}
We first show that if $\rho$ is absolutely PPT then $\alpha \leq \frac{2}{2+n^2}$. Notice that the eigenvalues $\lambda_1 \geq \lambda_2 \geq \cdots \geq \lambda_{n^2}$ of $\rho$ are 
\begin{equation*}
\begin{aligned}
\lambda_{1} = \alpha + \frac{1-\alpha}{n^2}, \quad \lambda_{i} = \frac{1-\alpha}{n^2} \quad \forall \, 2 \leq i \leq n^2.
\end{aligned}
\end{equation*}
Since we are assuming that $\rho$ is absolutely PPT, we can plug these eigenvalues into the LMIs described in \cite{Hil07} to obtain
\begin{equation*}
\begin{aligned}
 \frac{1}{n^2} \begin{bmatrix} 2 - 2 \alpha & -n^2\alpha & 0 & \cdots & 0 \\ 
																							-n^2\alpha & 2 - 2 \alpha & 0 & \cdots & 0 \\
																							0 & 0 & 2 - 2 \alpha & \cdots & 0 \\
																							\vdots & \vdots & \vdots & \ddots & \vdots \\
																							0 & 0 & 0 & \cdots & 2 - 2 \alpha
							  							\end{bmatrix} \geq 0.
\end{aligned}
\end{equation*}
The positivity of the above matrix implies that $n^2\alpha \leq 2 - 2\alpha$, so $\alpha \leq \frac{2}{2+n^2}$, as desired.

We now prove that if $\alpha \leq \frac{2}{2+n^2}$ then $\rho$ is absolutely separable. First notice that $U\rho U^\dagger = \frac{1-\alpha}{n^2}I + \alpha \ketbra{v}{v}$, where $\ket{v} := U\ket{\psi^+}$. If we absorb constants in a different way, we obtain
\begin{align}\label{eq:id_robustness}
	U\rho U^\dagger = \alpha\left(\frac{1-\alpha}{n^2\alpha}I + \ketbra{v}{v}\right),
\end{align}
which we want to show is separable for all $\ket{v}$.

To this end, we recall that it was shown in \cite{VT99} that the operator~\eqref{eq:id_robustness} is separable if and only if $\frac{1-\alpha}{n^2\alpha} \geq \gamma_1\gamma_2$, where $\gamma_1 \geq \gamma_2$ are the two largest Schmidt coefficients of $\ket{v}$. We can see that this is true for all $\ket{v}$ simply by noting that $\frac{1-\alpha}{n^2\alpha} \geq 1/2$ (which follows from the fact that $\alpha \leq \frac{2}{2+n^2}$), and $1/2 \geq \gamma_1\gamma_2$ (which can be seen by maximizing $\gamma_1\gamma_2$ subject to the constraint that $\gamma_1^2 + \gamma_2^2 \leq 1$). It follows that $\rho$ is absolutely separable, as desired.
\end{proof}

\subsection{UPB States}\label{sec:upb_states}

One of the more well-known methods of constructing PPT entangled states comes from the notion of an \emph{unextendible product basis (UPB)} \cite{BDMSST99}, which is a set of mutually orthogonal product states with no other product state orthogonal to all of them. Given a UPB $\mathcal{S} \subset \mathbb{C}^m \otimes \mathbb{C}^n$, it is straightforward to check that the following state is PPT and entangled:
\begin{align}\label{eq:upb_state}
	\rho = \frac{1}{mn - |\mathcal{S}|}\left(I - \sum_{\ket{v_i} \in \mathcal{S}} \ketbra{v_i}{v_i}\right).
\end{align}

States of the form~\eqref{eq:upb_state} cannot possibly be absolutely PPT since they are entangled yet have rank $mn - |\mathcal{S}| < mn$, which contradicts Proposition~\ref{prop:rank_of_ppt_spec}. However, we can follow the approach of \cite{BGR05} by letting $0 < p < 1$ be a real number and considering the full-rank state $\rho_p := pI/(mn) + (1-p)\rho$. As $p$ increases from $0$ to $1$, the state $\rho_p$ becomes arbitrarily close to the maximally-mixed state $I/(mn)$ and thus becomes absolutely PPT (and even absolutely separable) when $p$ is large enough. We now investigate the question of whether or not there exists $\rho_p$ that is absolutely PPT but not absolutely separable.
\begin{theorem}\label{thm:upb_sep_from_spec}
	Let $\rho \in M_3 \otimes M_3$ be a state constructed via a UPB, as in Equation~\eqref{eq:upb_state}. Then $\rho_p$ is absolutely PPT if and only if $p \geq 9(10-\sqrt{17})/83 \approx 0.6373\ldots$ Furthermore, $\rho_p$ is absolutely separable if $p \geq 1 - 1/\sqrt{10} \approx 0.6838\ldots$
\end{theorem}
\begin{proof}
	We first show that $\rho_p$ is absolutely PPT if and only if $p \geq 9(10-\sqrt{17})/83$. We first recall that all UPBs in $\mathbb{C}^3 \otimes \mathbb{C}^3$ have five states \cite{DMSST03}, so $\rho_p$ has eigenvalues $p/9$ with multiplicity $5$ and $p/9 + (1-p)/4 = (9-5p)/36$ with multiplicity $4$. It follows that the two LMIs determining whether or not $\rho_p$ is absolutely PPT are both as follows:
	\begin{align}\label{eq:upb_rho_lmi}
		\frac{1}{36}\begin{bmatrix}
			8p & 9p-9 & 9p-9 \\
			9p-9 & 8p & 9p-9 \\
			9p-9 & 9p-9 & 18-10p \\
		\end{bmatrix} \geq 0.
	\end{align}
	It is straightforward to check that the LMI~\eqref{eq:upb_rho_lmi} holds if and only if $p \geq 9(10-\sqrt{17})/83$, as claimed.
	
	To see that $\rho_p$ is absolutely separable when $p \geq 1 - 1/\sqrt{10}$, note that it suffices to consider the $p = 1 - 1/\sqrt{10}$ case. Define the operator $X := 8\rho_p = 8(1 - 1/\sqrt{10})I/9 + 8\rho/\sqrt{10}$. Then
	\begin{align*}
		\|X-I\|_F^2 & = \|8\rho/\sqrt{10} - (1 + 8/\sqrt{10})I/9\|_F^2 \\
		& = 4\big(2/\sqrt{10} - (1 + 8/\sqrt{10})/9\big)^2 + 5\big((1 + 8/\sqrt{10})/9\big)^2 \\
		& = 1.
	\end{align*}
	It follows from \cite[Theorem~1]{GB02} that $X$ is absolutely separable, so $\rho_p$ is too.
\end{proof}

From Theorem~\ref{thm:upb_sep_from_spec} we see that there is an interval approximately equal to $[0.6373,0.6838)$ for which $\rho_p$ is absolutely PPT, but we do not know whether or not it is absolutely separable. In the $p = 0.6373$ case, we have tried to detect entanglement in the state $U\rho_p U^\dagger$ for $10^5$ randomly-generated (according to Haar measure) unitary matrices $U$. We used every entanglement criterion that is known to us, including the tests based on covariance matrices \cite{GGHE08} and the extremely strong tests based on the $3$-copy PPT symmetric extensions \cite{DPS04}, and no entanglement was ever detected in any of these states. These numerical results seem to suggest that the state $\rho_p$ is absolutely separable when $p = 0.6373$ (and thus for all $p \in [0.6373,0.6838)$).

\section{Outlook} \label{sec:outlook}

We have investigated the problem of detecting entanglement within absolutely PPT states. We have provided a general technique, in the form of Lemma~\ref{lem:eigs_imp_sfs}, that can be used to prove that a given positive map or entanglement witness cannot detect entanglement in any absolutely PPT state. We have successfully applied this technique to the Choi map and its generalization, the Breuer--Hall map, and even the realignment criterion, but it could also be applied to other positive maps and entanglement witnesses in the literature (see \cite{CS14}, for example).

Certainly the most notable open problem resulting from this work is to answer Question~\ref{ques:main}, but other interesting questions also arise from our work. For instance, are there other separability criteria that can be shown to be incapable of detecting entanglement in absolutely PPT states? There are separability criteria such as those based on covariance matrices \cite{GGHE08} and symmetric extensions \cite{DPS04} for which we still do not know the answer.

Other open problems in this work include proving that \emph{all} of the generalized Choi maps $\Phi_{b,c}$ in Section~\ref{sec:choi_generalize} are incapable of detecting absolutely PPT entanglement (including those in the light gray region of Figure~\ref{fig:phi_bc}), and proving that the UPB states of Section~\ref{sec:upb_states} are absolutely separable when $p \in [0.6373,0.6838)$. We have provided numerical evidence that both of these claims are true, but we have been unable to find an analytic proof of either fact.

It would also be beneficial to develop other easily-checkable conditions that imply that a given entanglement witness cannot detect entanglement in absolutely PPT states. Almost all of our results follow from Lemma~\ref{lem:eigs_imp_sfs}, which gives a test based on the witness's maximal eigenvalue and the sum of its negative eigenvalues. However, it suffers from the drawback that the function $f$ (see Figure~\ref{fig:intervals}) is not continuous, which makes it difficult to use sometimes. Are there other functions of the eigenvalues of an entanglement witness that can be used in its place?

{\bf Acknowledgements.} The authors thank John Watrous for numerous insightful discussions and suggestions, and the proof of Lemma~\ref{lem:wit_trace_bound}. The authors also thank Ronald de Wolf for helpful comments on the writing of this paper, and two anonymous referees who caught errors in an early version of this paper. S.A. started working on this project during his studies at the Institute for Quantum Computing at the University of Waterloo and is grateful to Michele Mosca for support from NSERC. N.J. thanks Leonel Robert for some helpful conversations about the set of absolutely separable states and acknowledges support from NSERC. V.R. acknowledges support from NSERC and the US ARO.   

\bibliographystyle{ieeetr}
\bibliography{sfs}

\begin{thebibliography}{10}

\bibitem{W89}
R.~F. Werner, ``Quantum states with {E}instein-{P}odolsky-{R}osen correlations
  admitting a hidden-variable model,'' {\em Phys. Rev. A}, vol.~40,
  pp.~4277--4281, 1989.

\bibitem{GT09}
O.~G\"{u}hne and G.~Toth, ``Entanglement detection,'' {\em Phys. Rep.},
  vol.~474, pp.~1--75, 2009.

\bibitem{HHH09}
R.~Horodecki, P.~Horodecki, M.~Horodecki, and K.~Horodecki, ``Quantum
  entanglement,'' {\em Rev. Mod. Phys.}, vol.~81, pp.~865--942, 2009.

\bibitem{G03}
L.~Gurvits, ``Classical deterministic complexity of {E}dmonds' problem and
  quantum entanglement,'' in {\em Proceedings of the Thirty-Fifth Annual ACM
  Symposium on Theory of Computing}, pp.~10--19, 2003.

\bibitem{G10}
S.~Gharibian, ``Strong {NP}-hardness of the quantum separability problem,''
  {\em Quantum Inf. Comput.}, vol.~10, pp.~343--360, 2010.

\bibitem{P96}
A.~Peres, ``Separability criterion for density matrices,'' {\em Phys. Rev.
  Lett.}, vol.~77, pp.~1413--1415, 1996.

\bibitem{S63}
E.~St{\o}rmer, ``Positive linear maps of operator algebras,'' {\em Acta Math.},
  vol.~110, pp.~233--278, 1963.

\bibitem{W76}
S.~L. Woronowicz, ``Positive maps of low dimensional matrix algebras,'' {\em
  Rep. Math. Phys.}, vol.~10, pp.~165--183, 1976.

\bibitem{HHH96}
M.~Horodecki, P.~Horodecki, and R.~Horodecki, ``Separability of mixed states:
  Necessary and sufficient conditions,'' {\em Phys. Lett. A}, vol.~223,
  pp.~1--8, 1996.

\bibitem{KZ00}
M.~Ku{\'s} and K.~{\.Z}yczkowski, ``Geometry of entangled states,'' {\em Phys.
  Rev. A}, vol.~63, p.~032307, 2001.

\bibitem{OpenProb15}
E.~Knill, ``Separability from spectrum.'' Published electronically at
  \href{http://qig.itp.uni-hannover.de/qiproblems/15}{http://qig.itp.uni-hannover.de/qiproblems/15},
  2003.

\bibitem{VAD01}
F.~Verstraete, K.~Audenaert, and B.~D. Moor, ``Maximally entangled mixed states
  of two qubits,'' {\em Phys. Rev. A}, vol.~64, p.~012316, 2001.

\bibitem{EAO02}
A.~K. Ekert, C.~M. Alves, D.~K.~L. Oi, M.~Horodecki, P.~Horodecki, and L.~C.
  Kwek, ``Direct estimations of linear and nonlinear functionals of a quantum
  state,'' {\em Phys. Rev. Lett.}, vol.~88, p.~217901, 2002.

\bibitem{TOKKNN13}
T.~Tanaka, Y.~Ota, M.~Kanazawa, G.~Kimura, H.~Nakazato, and F.~Nori,
  ``Determining eigenvalues of a density matrix with minimal information in a
  single experimental setting,'' {\em Phys. Rev. A}, vol.~89, p.~012117, 2014.

\bibitem{GB02}
L.~Gurvits and H.~Barnum, ``Largest separable balls around the maximally mixed
  bipartite quantum state,'' {\em Phys. Rev. A}, vol.~66, p.~062311, 2002.

\bibitem{GCM14}
N.~Ganguly, J.~Chatterjee, and A.~S. Majumdar, ``Witness of mixed separable
  states useful for entanglement creation,'' {\em Phys. Rev. A}, vol.~89,
  p.~052304, 2014.

\bibitem{Hil07}
R.~Hildebrand, ``Positive partial transpose from spectra,'' {\em Phys. Rev. A},
  vol.~76, p.~052325, 2007.

\bibitem{JLNR14}
M.~A. Jivulescu, N.~Lupa, I.~Nechita, and D.~Reeb, ``Positive reduction from
  spectra.'' E-print: \href{http://arxiv.org/abs/1406.1277}{arXiv:1406.1277}
  [quant-ph], 2014.

\bibitem{Joh13}
N.~Johnston, ``Separability from spectrum for qubit--qudit states,'' {\em Phys.
  Rev. A}, vol.~88, p.~062330, 2013.

\bibitem{CW03}
K.~Chen and L.-A. Wu, ``A matrix realignment method for recognizing
  entanglement,'' {\em Quantum Inf. Comput.}, vol.~3, pp.~193--202, 2003.

\bibitem{R03}
O.~Rudolph, ``Some properties of the computable cross norm criterion for
  separability,'' {\em Phys. Rev. A}, vol.~67, p.~032312, 2003.

\bibitem{GLS93}
M.~Gr\"{o}tschel, L.~Lov\'{a}sz, and A.~Schrijver, {\em Geometric algorithms
  and combinatorial optimization}.
\newblock Springer-Verlag, 1993.

\bibitem{VB94}
L.~Vandenberghe and S.~Boyd, ``Semidefinite programming,'' {\em SIAM Review},
  vol.~38, pp.~49--95, 1996.

\bibitem{Wat04Lec7}
J.~Watrous, ``Theory of quantum information lecture notes, lecture 7.''
  published electronically at
  \href{http://www.cs.uwaterloo.ca/~watrous/lecture-notes.html}{http://www.cs.uwaterloo.ca/~watrous/lecture-notes.html},
  2004.

\bibitem{C75}
M.-D. Choi, ``Completely positive linear maps on complex matrices,'' {\em
  Linear Algebra Appl.}, vol.~10, pp.~285--290, 1975.

\bibitem{oeisA237749}
N.~Johnston, ``The {O}n-{L}ine {E}ncyclopedia of {I}nteger {S}equences.''
  \href{http://oeis.org/A237749}{A237749}, February 2014.
\newblock The number of possible orderings of the real numbers $x_i x_j$ ($i
  \leq j$), subject to the constraint that $x_1 > x_2 > ... > x_n > 0$.

\bibitem{Bha97}
R.~Bhatia, {\em Matrix analysis}.
\newblock Springer, 1997.

\bibitem{H97}
P.~Horodecki, ``Separability criterion and inseparable mixed states with
  positive partial transposition,'' {\em Phys. Lett. A}, vol.~232,
  pp.~333--339, 1997.

\bibitem{BP00}
D.~Bru\ss and A.~Peres, ``Construction of quantum states with bound
  entanglement,'' {\em Phys. Rev. A}, vol.~61, p.~030301(R), 2000.

\bibitem{BDMSST99}
C.~H. Bennett, D.~P. DiVincenzo, T.~Mor, P.~W. Shor, J.~A. Smolin, and B.~M.
  Terhal, ``Unextendible product bases and bound entanglement,'' {\em Phys.
  Rev. Lett.}, vol.~82, pp.~5385--5388, 1999.

\bibitem{HHH98}
M.~Horodecki, P.~Horodecki, and R.~Horodecki, ``Mixed-state entanglement and
  distillation: Is there a ``bound'' entanglement in nature?,'' {\em Phys. Rev.
  Lett.}, vol.~80, pp.~5239--5242, 1998.

\bibitem{BGR05}
S.~Bandyopadhyay, S.~Ghosh, and V.~Roychowdhury, ``Non-full rank bound
  entangled states satisfying the range criterion,'' {\em Phys. Rev. A},
  vol.~71, p.~012316, 2005.

\bibitem{CD14}
L.~Chen and D.~\v{Z}. {\DH}okovi\'{c}, ``Boundary of the set of separable
  states.'' E-print: \href{http://arxiv.org/abs/1404.0738}{arXiv:1404.0738}
  [quant-ph], 2014.

\bibitem{CKL92}
S.-J. Cho, S.-H. Kye, and S.~G. Lee, ``Generalized {C}hoi maps in
  $3$-dimensional matrix algebras,'' {\em Linear Algebra Appl.}, vol.~171,
  pp.~213--224, 1992.

\bibitem{Bre06}
H.-P. Breuer, ``Optimal entanglement criterion for mixed quantum states,'' {\em
  Phys. Rev. Lett.}, vol.~97, p.~080501, 2006.

\bibitem{Hal06}
W.~Hall, ``A new criterion for indecomposability of positive maps,'' {\em J.
  Phys. A: Math. Gen.}, vol.~39, p.~14119, 2006.

\bibitem{cvx}
M.~Grant and S.~Boyd, ``{CVX}: {MATLAB} software for disciplined convex
  programming, version 2.0 beta.''
  \href{http://cvxr.com/cvx}{http://cvxr.com/cvx}, Sept. 2012.

\bibitem{AJR14code}
A.~Arunachalam, N.~Johnston, and V.~Russo, ``Supplementary scripts for
  implementing separability from spectrum {SDP}s numerically.''
  \href{https://bitbucket.org/vprusso/sep-from-spec-software}{https://bitbucket.org/vprusso/sep-from-spec-software},
  2014.

\bibitem{Joh12}
N.~Johnston, {\em Norms and Cones in the Theory of Quantum Entanglement}.
\newblock PhD thesis, University of Guelph, 2012.

\bibitem{Cho80}
M.-D. Choi, ``Some assorted inequalities for positive linear maps on
  ${C}^*$-algebras,'' {\em J. Operator Theory}, vol.~4, pp.~271--285, 1980.

\bibitem{Ha13}
K.-C. Ha, ``Notes on extremality of the {Choi} map,'' {\em Linear Algebra
  Appl.}, vol.~439, pp.~3156--3165, 2013.

\bibitem{Kit97}
A.~Y. Kitaev, ``Quantum computations: Algorithms and error correction,'' {\em
  Russian Math. Surveys}, vol.~52, pp.~1191--1249, 1997.

\bibitem{Wa09}
J.~Watrous, ``Semidefinite programs for completely bounded norms,'' {\em Theory
  Comput.}, vol.~5, pp.~217--238, 2009.

\bibitem{Wat12}
J.~Watrous, ``Simpler semidefinite programs for completely bounded norms.''
  E-print: \href{http://arxiv.org/abs/1207.5726}{arXiv:1207.5726} [quant-ph],
  2012.

\bibitem{Wat05}
J.~Watrous, ``Notes on super-operator norms induced by {S}chatten norms,'' {\em
  Quantum Inf. Comput.}, vol.~5, pp.~58--68, 2005.

\bibitem{NOP09}
M.~Navascu\'{e}s, M.~Owari, and M.~B. Plenio, ``Power of symmetric extensions
  for entanglement detection,'' {\em Phys. Rev. A}, vol.~80, p.~052306, 2009.

\bibitem{JK11}
N.~Johnston and D.~W. Kribs, ``A family of norms with applications in quantum
  information theory {II},'' {\em Quantum Inf. Comput.}, vol.~11, pp.~104--123,
  2011.

\bibitem{CW11}
D.~Chru{\'s}ci{\'n}ski and F.~A. Wudarski, ``Geometry of entanglement witnesses
  for two qutrits,'' {\em Open Systems \& Information Dynamics}, vol.~18,
  no.~04, pp.~375--387, 2011.

\bibitem{CS13}
D.~Chru{\'s}ci{\'n}ski and G.~Sarbicki, ``Optimal entanglement witnesses for
  two qutrits,'' {\em Open Systems \& Information Dynamics}, vol.~20, no.~02,
  p.~1350006, 2013.

\bibitem{HH99}
M.~Horodecki and P.~Horodecki, ``Reduction criterion of separability and limits
  for a class of distillation protocols,'' {\em Phys. Rev. A}, vol.~59,
  pp.~4206--4216, 1999.

\bibitem{HK11}
K.-C. Ha and S.-H. Kye, ``Entanglement witnesses arising from exposed positive
  linear maps,'' {\em Open Syst. Inf. Dyn.}, vol.~18, pp.~323--337, 2011.

\bibitem{MMX10}
C.~Mehl, V.~Mehrmann, and H.~Xu, ``Singular-value-like decomposition for
  complex matrix triples,'' {\em Journal of computational and applied
  mathematics}, vol.~233, no.~5, pp.~1245--1276, 2010.

\bibitem{VT99}
G.~Vidal and R.~Tarrach, ``Robustness of entanglement,'' {\em Phys. Rev. A},
  vol.~59, pp.~141--155, 1999.

\bibitem{DMSST03}
D.~P. DiVincenzo, T.~Mor, P.~W. Shor, J.~A. Smolin, and B.~M. Terhal,
  ``Unextendible product bases, uncompletable product bases and bound
  entanglement,'' {\em Commun. Math. Phys.}, vol.~238, pp.~379--410, 2003.

\bibitem{GGHE08}
O.~Gittsovich, O.~G\"uhne, P.~Hyllus, and J.~Eisert, ``Unifying several
  separability conditions using the covariance matrix criterion,'' {\em Phys.
  Rev. A}, vol.~78, p.~052319, 2008.

\bibitem{DPS04}
A.~C. Doherty, P.~A. Parrilo, and F.~M. Spedalieri, ``A complete family of
  separability criteria,'' {\em Phys. Rev. A}, vol.~69, p.~022308, 2004.

\bibitem{CS14}
D.~Chru\'{s}ci\'{n}ski and G.~Sarbicki, ``Entanglement witnesses: construction,
  analysis and classification.'' E-print:
  \href{http://arxiv.org/abs/1402.2413}{arXiv:1402.2413} [quant-ph], 2011.

\end{thebibliography}

\section{Appendix A. Proof of Theorem~\ref{thm:gen_choi_from_spec}}\label{sec:gen_choi_ub}

We now prove some upper bounds on the eigenvalues of $(id_3 \otimes \Phi_{b,c}^\dagger)(\ketbra{v}{v})$, which allow us to prove Theorem~\ref{thm:gen_choi_from_spec}. We note that these upper bounds are much more complicated than the lower bound given by Lemma~\ref{lem:choi_eigs_gen}, so the upper bounds are illustrated in Figure~\ref{fig:gen_choi_ub} for clarity.
\begin{figure}[htb]
	\centering
\begin{tikzpicture}
\begin{axis}[
axis equal,
xlabel=$b$,
ylabel=$c$,
xmin=0, xmax=1.3333333333333333,
ymin=-0.16666666666, ymax=1.5,
xtick={0,0.3333333333333333,0.6666666666667,1,1.3333333333333333},
ytick={0,0.3333333333333333,0.6666666666667,1,1.3333333333333333},
xticklabels={$0$,$\tfrac{1}{3}$,$\tfrac{2}{3}$,$1$,$\tfrac{4}{3}$},
yticklabels={$0$,$\tfrac{1}{3}$,$\tfrac{2}{3}$,$1$,$\tfrac{4}{3}$},
grid=major,
]

\addplot [forget plot] graphics [xmin=0,xmax=1.3333333333,ymin=0,ymax=1.3333333333] {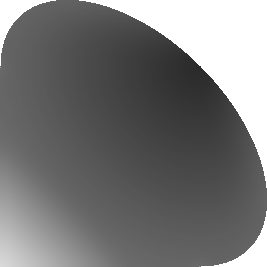};

\draw[domain=0:1,smooth,variable=\t,black,line width=1.5] plot ({100*\t*\t/(1-\t+\t*\t)},{100*(0.16666666666+1/(1-\t+\t*\t))});
\draw[domain=0:1,smooth,variable=\t,black,line width=1.5] plot ({100/(1-\t+\t*\t)},{100*(0.16666666666+\t*\t/(1-\t+\t*\t))});
\addplot[black,line width=0.5] coordinates {(1,0) (0,0) (0,1)};

\addplot[black,line width=0.5,dashed] coordinates {(0.3333,1.3333) (1,1) (1.3333,0.3333)};

\addplot[black,fill=white] (1,0) circle [radius=0.015];

\addplot[black,fill=white] (0,1) circle [radius=0.015];

\addplot[black,fill=white] (1,1) circle [radius=0.015];

\addplot[black,fill=white] (1.3333,0.3333) circle [radius=0.015];

\addplot[black,fill=white] (0.3333,1.3333) circle [radius=0.015];

\addplot[black,line width=0.5] coordinates {(1.25,0.5) (1.2,0.475)};
\node at (axis cs:1.2,0.475) [anchor=north east] {$2b + c = 3$};

\addplot[black,line width=0.5] coordinates {(0.5,1.25) (0.475,1.2)};
\node at (axis cs:0.6,1.2) [anchor=north east] {$b + 2c = 3$};
\end{axis}

\node[align=right] at (0.8,-1.6) {Best eigenvalue \\ upper bound};

\shade[shading=axis,bottom color=black!15!white,top color=black!83!white,shading angle=90] (2.45,-1.8) rectangle (7.4,-1.4);
\draw [black,line width=0.5] (2.45,-1.8) rectangle (7.4,-1.4);

\node[align=center] at (2.45,-2.12) {$\tfrac{1}{2}$};
\node[align=center] at (4.1,-2.12) {$\tfrac{2}{3}$};
\node[align=center] at (5.75,-2.12) {$\tfrac{5}{6}$};
\node[align=center] at (7.4,-2.12) {$1$};
\end{tikzpicture}
\caption{A plot of the best possible upper bound on the eigenvalues of $(id_3 \otimes \Phi_{b,c}^\dagger)(\ketbra{v}{v})$, which was computed numerically using the semidefinite program~\eqref{sdp:eig_ub}. This bound agrees with the bounds provided by Lemma~\ref{lem:choi_eigs_gen_ub1} whenever $b+c \geq \frac{2}{3}$. The area above and to the right of the two dashed lines is covered by the first case of that lemma, whereas the area below and to the left of the dashed lines is covered by the second case of the lemma.}\label{fig:gen_choi_ub}
\end{figure}
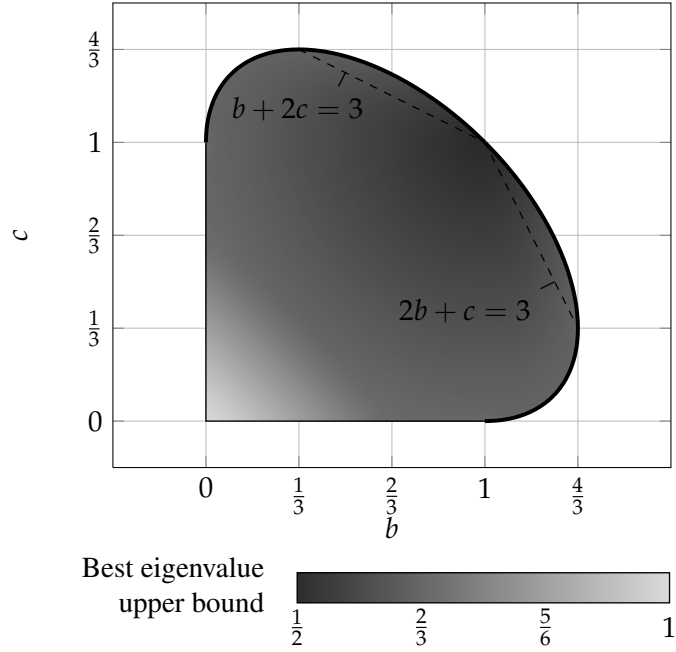

\begin{lemma}\label{lem:choi_eigs_gen_ub1}
	Let $\ket{v} \in \mathbb{C}^3 \otimes \mathbb{C}^3$ and suppose $b,c \geq 0$ are such that $b+c \geq \frac{2}{3}$. We split into two cases:
	\begin{itemize}
		\item If $2b + c \geq 3$ or $b + 2c \geq 3$, then the eigenvalues of $(id_3 \otimes \Phi_{b,c}^\dagger)(\ketbra{v}{v})$ are no larger than $\max\{b,c\}/2$.
		\item Otherwise, the eigenvalues of $(id_3 \otimes \Phi_{b,c}^\dagger)(\ketbra{v}{v})$ are no larger than
		\begin{align*}
			\frac{b^2+c^2-6(b+c)+bc+9}{6(2-b-c)}.
		\end{align*}
	\end{itemize}
\end{lemma}
\begin{proof}
	We can get an upper bound of the eigenvalues of $(id_3 \otimes \Phi_{b,c}^\dagger)(\ketbra{v}{v})$ in a manner similar to that which was used in the proof of Lemma~\ref{lem:choi_eigs}. In particular, it suffices to give a feasible point for the dual program of the SDP~\eqref{sdp:eig_ub}, with $\Phi_C$ replaced by $\Phi_{b,c}$.
	
	To this end, we start by considering the first case (i.e., the case where $2b + c \geq 3$ or $b + 2c \geq 3$). It suffices to take $Y = 0$ in the dual program of the SDP~\eqref{sdp:eig_ub}. Then we can write $J(\Phi_{b,c}^\dagger)$ in the standard basis of $\mathbb{C}^3 \otimes \mathbb{C}^3$ as follows:
	\begin{align*}
		\frac{1}{2}\begin{bmatrix}
			2-b-c &  \cdot &  \cdot &  \cdot &  -1 &  \cdot &  \cdot &  \cdot & -1 \\
			\cdot &  b &  \cdot &  \cdot &  \cdot &  \cdot &  \cdot &  \cdot &  \cdot \\
	 		\cdot &  \cdot & c &  \cdot &  \cdot &  \cdot &  \cdot &  \cdot &  \cdot \\
	 		\cdot &  \cdot &  \cdot & c &  \cdot &  \cdot &  \cdot &  \cdot &  \cdot \\
	 		-1 &  \cdot &  \cdot &  \cdot &  2-b-c &  \cdot &  \cdot &  \cdot &  -1 \\
	 		\cdot &  \cdot &  \cdot &  \cdot &  \cdot &  b &  \cdot &  \cdot &  \cdot \\
	 		\cdot &  \cdot &  \cdot &  \cdot &  \cdot &  \cdot &  b &  \cdot &  \cdot \\
			\cdot &  \cdot &  \cdot &  \cdot &  \cdot &  \cdot &  \cdot & c &  \cdot \\
	 		-1 &  \cdot &  \cdot &  \cdot &  -1 &  \cdot &  \cdot &  \cdot &  2-b-c
		\end{bmatrix},
	\end{align*}
	from which it follows that $\lambda_{\textup{max}}\big(J(\Phi_{C}^\dagger)\big) = \max\{b,c,3-b-c\}/2$. Since $2b + c \geq 3$ or $b + 2c \geq 3$ it follows that $\max\{b,c,3-b-c\}/2 = \max\{b,c\}/2$, which shows that the SDP~\eqref{sdp:eig_ub} has optimal value no larger than $\max\{b,c\}/2$ and competes the proof of this case.
	
	We now consider the other case (i.e., we assume that $2b + c < 3$ and $b + 2c < 3$). Define the two quantities $x := (3 - 2b - c)^2/(6(2-b-c))$ and $y := (3 - b - 2c)^2/(6(2-b-c))$. Then consider the following operator, written with respect to the standard basis of $\mathbb{C}^3 \otimes \mathbb{C}^3$:
	\begin{align*}
		Y & = \begin{bmatrix}
			\cdot &  \cdot &  \cdot &  \cdot &  \cdot &  \cdot &  \cdot &  \cdot &  \cdot \\
			\cdot &  x &  \cdot &  \sqrt{xy} &  \cdot &  \cdot &  \cdot &  \cdot &  \cdot \\
	 		\cdot &  \cdot & y &  \cdot &  \cdot &  \cdot &  \sqrt{xy} &  \cdot &  \cdot \\
	 		\cdot &  \sqrt{xy} &  \cdot & y &  \cdot &  \cdot &  \cdot &  \cdot &  \cdot \\
	 		\cdot &  \cdot &  \cdot &  \cdot &  \cdot &  \cdot &  \cdot &  \cdot &  \cdot \\
	 		\cdot &  \cdot &  \cdot &  \cdot &  \cdot &  x &  \cdot &  \sqrt{xy} &  \cdot \\
	 		\cdot &  \cdot &  \sqrt{xy} &  \cdot &  \cdot &  \cdot &  x &  \cdot &  \cdot \\
			\cdot &  \cdot &  \cdot &  \cdot &  \cdot &  \sqrt{xy} &  \cdot & y &  \cdot \\
	 		\cdot &  \cdot &  \cdot &  \cdot &  \cdot &  \cdot &  \cdot &  \cdot &  \cdot
		\end{bmatrix}.
	\end{align*}
	It is straightforward to see that $Y \geq 0$ and is thus a feasible point of the SDP~\eqref{sdp:eig_ub}. To see what the corresponding value of the objective function is, we compute
	\begin{align*}
		&(id \otimes T)(Y) + J(\Phi_{b,c}^\dagger) = \\ 
		&\frac{1}{2}\begin{bmatrix}
			2-b-c &  \cdot &  \cdot &  \cdot &  2\sqrt{xy}-1 &  \cdot &  \cdot &  \cdot & 2\sqrt{xy}-1 \\
			\cdot &  b+2x &  \cdot &  \cdot &  \cdot &  \cdot &  \cdot &  \cdot &  \cdot \\
	 		\cdot &  \cdot & c+2y &  \cdot &  \cdot &  \cdot &  \cdot &  \cdot &  \cdot \\
	 		\cdot &  \cdot &  \cdot & c+2y &  \cdot &  \cdot &  \cdot &  \cdot &  \cdot \\
	 		2\sqrt{xy}-1 &  \cdot &  \cdot &  \cdot &  2-b-c &  \cdot &  \cdot &  \cdot & 2\sqrt{xy}-1 \\
	 		\cdot &  \cdot &  \cdot &  \cdot &  \cdot &  b+2x &  \cdot &  \cdot &  \cdot \\
	 		\cdot &  \cdot &  \cdot &  \cdot &  \cdot &  \cdot &  b+2x &  \cdot &  \cdot \\
			\cdot &  \cdot &  \cdot &  \cdot &  \cdot &  \cdot &  \cdot & c+2y &  \cdot \\
	 		2\sqrt{xy}-1 &  \cdot &  \cdot &  \cdot &  2\sqrt{xy}-1 &  \cdot &  \cdot &  \cdot & 2-b-c
		\end{bmatrix}.
	\end{align*}
	It is straightforward to verify that $b+2x = c+2y = 2-b-c - (2\sqrt{xy}-1)$ for the given choice of $x$ and $y$, from which it follows that
	\begin{align*}
		& \quad \ (id \otimes T)(Y) + J(\Phi_{b,c}^\dagger) \\
		& = \frac{1}{2}\Big((b+2x)I + 3(2\sqrt{xy}-1)\ketbra{\psi^+}{\psi^+}\Big).
	\end{align*}
	Since $2\sqrt{xy} \leq 1$ whenever $b$ and $c$ satisfy the constraints of this case, we see that the maximum eigenvalue of $(id \otimes T)(Y) + J(\Phi_{b,c}^\dagger)$ equals $b/2+x = \frac{b^2+c^2-6(b+c)+bc+9}{6(2-b-c)}$, so this quantity is an upper bound on the optimal value of the SDP~\eqref{sdp:eig_ub}, as desired.
\end{proof}

Now that we have these upper bounds of Lemma~\ref{lem:choi_eigs_gen_ub1} to work with, we are finally in a position to prove Theorem~\ref{thm:gen_choi_from_spec}.
\begin{proof}[Proof of Theorem~\ref{thm:gen_choi_from_spec}]
	We first note that it suffices to prove Theorem~\ref{thm:gen_choi_from_spec} for the four maps $\Phi_{b,c}$ given by the points $(a)$--$(d)$ that it describes, since the result then immediately follows for any convex combination of those maps. We consider these four maps now, one at a time.
	
	{\bf Case $(a)$:} $(b,c) = (0,0)$. It is straightforward to check that $\Phi_{0,0}$ is completely positive, so the result is trivial.
	
	For each of the remaining cases, we follow the notation of Lemma~\ref{lem:eigs_imp_sfs} and use $\ell$ to denote the sum of the negative eigenvalues of $(id_3 \otimes \Phi_{b,c}^\dagger)(\ketbra{v}{v})$, we use $\mu_1$ to denote its maximal eigenvalue, and we use $f$ to denote the function described by Lemma~\ref{lem:eigs_imp_sfs}. Furthermore, we note that $(id_3 \otimes \Phi_{b,c}^\dagger)(\ketbra{v}{v})$ has at most one negative eigenvalue whenever $b+c \leq 3$, so any lower bound on the eigenvalues of $(id_3 \otimes \Phi_{b,c}^\dagger)(\ketbra{v}{v})$ immediately applies to $\ell$ as well.
	
	{\bf Case $(b)$:} $(b,c) = (0, 3(\sqrt{2}-1))$. We know from Lemma~\ref{lem:choi_eigs_gen} that $\ell \geq -\tfrac{1}{6}(b+c) = \tfrac{1}{2}(1-\sqrt{2})$. Furthermore, plugging this choice of $b$ and $c$ into Lemma~\ref{lem:choi_eigs_gen_ub1} shows that $\mu_1 \leq \frac{1}{7}(9 - 3\sqrt{2})$. Since $\frac{1}{7}(9 - 3\sqrt{2}) \approx 0.6796\ldots \leq f(\tfrac{1}{2}(1-\sqrt{2})) = \tfrac{1}{4}(2+\sqrt{2}) \approx 0.8535\ldots$, the result follows from Lemma~\ref{lem:eigs_imp_sfs}.
	
	{\bf Case $(c)$:} $(b,c) = (6/5, 6/5)$. We know from Lemma~\ref{lem:choi_eigs_gen} that $\ell \geq -\tfrac{1}{6}(b+c) = -\tfrac{2}{5}$. Furthermore, plugging this choice of $b$ and $c$ into Lemma~\ref{lem:choi_eigs_gen_ub1} shows that $\mu_1 \leq \frac{3}{5}$. Since $\frac{3}{5} \leq f(-\tfrac{2}{5}) = \frac{3}{5}$, the result follows from Lemma~\ref{lem:eigs_imp_sfs}.
		
	{\bf Case $(d)$:} $(b,c) = (3(\sqrt{2}-1), 0)$. Almost identical to case~$(b)$.
\end{proof}
\end{document}